\documentclass[pra, two column]{revtex4-1}
\usepackage{graphicx, amsmath, amssymb, amsthm, bm, xcolor, CJK}
\usepackage[colorlinks=true,linkcolor=blue,urlcolor=blue,citecolor=blue]{hyperref}
\hypersetup{breaklinks=true}

\newtheorem{lemma}{Lemma}
\newtheorem{prop}{Proposition}
\newtheorem{corollary}{Corollary}

\DeclareMathOperator{\tr}{tr}

\begin{document}

\begin{CJK*}{UTF8}{}

\CJKfamily{gbsn}

\title{Information scrambling in chaotic systems with dissipation}
\author{Yong-Liang Zhang}
\email{ylzhang@caltech.edu}
\author{Yichen Huang (黄溢辰)}
\email{ychuang@caltech.edu}
\author{Xie Chen}
\email{xiechen@caltech.edu}

\affiliation{Department of Physics and Institute for Quantum Information and Matter, California Institute of Technology, Pasadena, California 91125, USA}

\date{\today}

\begin{abstract}
Chaotic dynamics in closed local quantum systems scrambles quantum information, which is manifested quantitatively in the decay of the out-of-time-ordered correlators (OTOC) of local operators. How is information scrambling affected when the system is coupled to the environment and suffers from dissipation? In this paper, we address this question by defining a dissipative version of OTOC and numerically study its behavior in a prototypical chaotic quantum chain in the presence of dissipation. We find that dissipation leads to not only the overall decay of the scrambled information due to leaking, but also structural changes so that the `information light cone' can only reach a finite distance even when the effect of overall decay is removed. Based on this observation we conjecture a modified version of the Lieb-Robinson bound in dissipative systems.
\end{abstract}
\maketitle

\end{CJK*}

\section{Introduction}
Chaos in quantum mechanical systems is characterized by the scrambling of quantum information. More specifically, suppose that information is encoded initially in a local operator $A$. Under the dynamics generated by a local Hamiltonian $H = \sum_{i} h_i$, $A(t) = e^{iHt} A e^{-iHt}$ grows in size and becomes non-local as $t$ increases. As $A$ grows in size, it starts to overlap with local operators $B$ at other spatial locations and ceases to commute with them. The effect of information scrambling is then manifested as the growth in the norm of the commutator $[A(t), B]$. Correspondingly it is also manifested as the decay of (the real part of) the out-of-time-ordered correlator (OTOC) \cite{LO69,Kit14,Kitaev,Shenker2014,Shenker201412,Shenker2015,Roberts:2014ab,Roberts2015,Roberts:2016aa,Hosur2016,Maldacena2016,polchinski2016spectrum,gu2016local,PRL118.086801,roberts2016chaos,patel2017,huang2017,Mezei2017,Lucas:2017aa,huang2016,FZSZ16,C16,SC16,HL16,chen2016,slagle2017out} $\langle A^{\dagger}(t)B^{\dagger}A(t)B\rangle_\beta$ which is related to the commutator as
\begin{equation}
\Re\,\langle A^{\dagger}(t)B^{\dagger}A(t)B\rangle_\beta = 1 - \frac{1}{2} \langle [A(t), B]^{\dagger}[A(t), B] \rangle_\beta,
\end{equation}
where local operators $A, B$ are both unitary, $\langle\cdot\rangle_\beta$ represents the thermal average at the inverse temperature $\beta=1/T$, and $\Re$ denotes the real part.

In a chaotic system, the decay of OTOC is usually expected to exhibit the following features: First, after time evolution for a very long time, information initially encoded in $A$ becomes highly nonlocal and cannot be accessed with any individual local operator $B$. Therefore, all OTOCs decay to zero at late time \cite{Roberts:2014ab,Hosur2016,roberts2016chaos,huang2017}
\begin{equation}
\lim_{t \to \infty} \Re\,\langle A^{\dagger}(t)B^{\dagger}A(t)B\rangle_{\beta=0} = 0.
\end{equation}

Secondly, in chaotic $0$-dimensional systems, the OTOC starts to decay at early time in an exponential way \cite{Hosur2016}
\begin{equation}
\Re\,\langle A^{\dagger}(t)B^{\dagger}A(t)B\rangle_\beta = f_1 - \frac{f_2}{n} e^{\lambda_L t} +O\left(\frac{1}{n^2}\right),
\end{equation}
where the constants $f_1, f_2$ depend on the choice of operators $A,B$, and $n$ is total number of degrees of freedom. The exponent of the exponential -- the Lyapunov exponent -- characterizes how chaotic the quantum dynamics is. It is bounded by $\lambda_L \leq \frac{2\pi}{\beta}$ \cite{Maldacena2016,Roberts:2016aa,Hosur2016,polchinski2016spectrum} and is expected to be saturated by quantum systems corresponding to black holes. 

Thirdly, in a system with spatial locality, information spreads at a certain speed, giving rise to a delay time before OTOC starts to decay. In some simple cases \cite{Shenker2014,Hosur2016,gu2016local,Song2017,Ben-Zion2017}, the early-time behavior of OTOC is described by
\begin{equation} \label{butterfly_velocity}
f'_1 - f'_2 e^{\lambda_L (t- d_{BA}/v_B)} +O(e^{-2\lambda_L d_{BA}/v_B})    
\end{equation}
with some constants $f'_1, f'_2$ that depend on $A, B$, and $d_{BA}$ is the distance between the local operators $A$ and $B$. That is, information spreads with a finite velocity $v_B$ -- the butterfly velocity -- and forms a `light cone' \cite{Roberts2015,Roberts:2016aa,Hosur2016}. In more general systems, the wave front of the light cone becomes wider while propagating out and Ref. \cite{Xu2018} gives an in depth study of the general form of the early time decay of OTOC. The deep connection between OTOC and quantum chaos generated a lot of interest in the topic, both theoretically and experimentally. Several protocols have been proposed to measure these unconventional correlators in real experimental systems \cite{Swingle:2016aa,ZHG16,YGS+16,tsuji2017exact,bohrdt2016,garttner2016,li2016measuring,Halpern:2017ab,Halpern:2017aa}. 

The measurement of OTOC in real experimental systems is complicated by the fact that the system is not exactly closed and suffers from dissipation through coupling to the environment. How does dissipation affect the measured signal of OTOC? More generally, we can ask how does dissipation affect information scrambling in a chaotic system? Dissipation leads to leakage of information, and therefore it is natural to expect that any signal of information scrambling would decay. Is it then possible to recover the signatures of information scrambling in a dissipative system and observe the existence of a light cone?

We address this question by studying numerically a prototypical model of chaotic spin chain \cite{PRL106.050405,Shenker2014,Roberts2015,Hosur2016} -- the Ising model with both transverse and longitudinal fields -- in the presence of some common types of dissipation: amplitude damping, phase damping and phase depolarizing. The Hamiltonian of the system with open boundary condition is 
\begin{equation} \label{ising}
H_s= - \sum_{i=1}^{N-1} \sigma^z_i \sigma^z_{i+1} - \sum_{i=1}^N (g \sigma^x_i + h \sigma^z_i),    
\end{equation}
where we set the parameters to be $g=-1.05$ and $ h=0.5$, and $N$ is the number of spins. We find that if OTOC is measured using the protocol given in Ref. \cite{Swingle:2016aa}, dissipation leads to the decay of the signal not only due to information leaking into the environment, but also information re-structuring. We define a corrected OTOC to remove the effect of leaking, so that the light cone can be recovered to some extent. However, due to the re-structuring, the recovered light cone only persists to a finite distance. 

The paper is organized as follows. In Sec. \ref{sec2}, we review the dynamics of dissipative systems and define a dissipative version of OTOC based on the measurement protocol given in Ref. \cite{Swingle:2016aa}. In Sec. \ref{sec3}, after observing the fast overall decay of the dissipative OTOC, we define a corrected OTOC to remove the effect of overall information leaking in the hope of recovering the information light cone. However, we see that the corrected light cone still only persists for a finite distance. In Sec. \ref{sec4}, we point out that the corrected light cone is finite due to information re-structuring and investigate the relationship between the width of the partially recovered light cone and the strength of dissipation. In Sec. \ref{sec5}, we conjecture a modified Lieb-Robinson bound for dissipative systems based on our observation regarding OTOC in the previous sections.

\section{Measurement of OTOC in dissipative systems}\label{sec2}
In this section, we provide a brief review of the dynamics of dissipative systems, and then generalize the definition of OTOC to dissipative systems based on the measurement protocol in Ref. \cite{Swingle:2016aa}.

A dissipative system is an open quantum system $S$ coupled to its environment $E$. In this coupled system, the total Hamiltonian is $H = H_{s} + H_{e} + H_{int}$, where $H_s(H_e)$ is the Hamiltonian of the system (environment) and $H_{int}$ is the interaction term. The reduced density matrix of the system $S$ changes as a consequence of its internal dynamics and the interaction with the environment $E$. In most cases, the initial state is assumed to be a product state $\rho_s(0)\otimes\rho_e(0)$. Under the Born, Markov and secular approximations, the dynamical evolution of a dissipative system $\rho_s(t) = \tr_e[ e^{-iHt}\rho_s(0)\otimes\rho_e (0) e^{iHt}] = \mathcal{V}(t) \cdot \rho_s (0)$ can be described by the Lindblad master equation \cite{breuer2002}
\begin{align}
&\frac{d \rho_s(t)}{d t} = \mathcal{L} \cdot \rho_s(t) = -i[H_s, \rho_s(t)] + \nonumber \\
&\sum_k \frac{\Gamma }{2}\Big( 2L_k\rho_s(t)L_k^\dag- \rho_s(t)L_k^\dag L_k - L^\dag_k L_k \rho_s(t) \Big),
\label{master eq}
\end{align}
where the first commutator with $H_s$ represents the unitary dynamics, the dissipation rate $\Gamma$ is a positive number, the Lindblad operators $L_k$ describe the dissipation, and $\mathcal{L}$ is the Liouvillian super-operator. Some common types of dissipation \cite{breuer2002,zeng2015} act locally on each spin via the Lindblad operators:
\begin{align}
\text{amplitude damping:}\ \  & L_k = \sqrt{\frac{1}{2}} (\sigma^x_k - i\sigma^y_k), \\
\text{phase damping:}\ \  & L_k = \sqrt{\frac{1}{2}}\sigma^z_k, \\ 
\text{phase depolarizing:}\ \  & L_k = \frac{1}{2} \sigma^x_k, \frac{1}{2} \sigma^y_k, \frac{1}{2} \sigma^z_k,
\end{align}
where $k$ denotes the $k$-th spin.

In the Heisenberg picture, the adjoint dynamical map $\mathcal{V}^\dag(t)$ acting on the Hermitian operators is defined by $ \tr[O (\mathcal{V}(t) \cdot \rho_s)] = \tr[(\mathcal{V}^\dag(t)\cdot O) \rho_s]$ for all states $\rho_s$. If the Lindblad operators do not depend on time, then the adjoint master equation describing the evolution of the operator $O_H(t) = \mathcal{V}^\dag(t) \cdot O$ is \cite{breuer2002}
\begin{align}
&\frac{d O_H(t)}{d t} = \mathcal{L}^\dag \cdot O_H(t) = i[H_s, O_H(t)] +\nonumber \\
& \sum_k \frac{\Gamma}{2}\Big( 2L_k^\dag O_H(t)L_k - O_H(t)L_k^\dag L_k - L^\dag_k L_k O_H(t) \Big).
\label{adjoint master eq}
\end{align}

Given both the dynamical and the adjoint dynamical map, how should we define the OTOC in a dissipative system? Should we just replace $A(t)$ with $\mathcal{V}^\dag(t) \cdot A$ or do something more complicated? In order to give a meaningful answer to this question, we need to specialize to a particular measurement scheme of OTOC and see how the measured quantity changes due to dissipation. We choose to focus on the measurement scheme given in Ref. \cite{Swingle:2016aa}. 

Let us analyze in more detail how the measurement scheme would be affected if dissipation is present. Without dissipation, the protocol involves the system whose unitary dynamics generated by $H_s$ is to be probed and a control qubit $c$. The system is initialized in a thermal state $\rho_s$ or eigenstate $|\psi\rangle_s$ and the control qubit is initialized in state $|+\rangle_c = \frac{1}{\sqrt{2}}(|0\rangle_c + |1\rangle_c)$. Ignoring dissipation, the measurement scheme involves the following steps of unitary operations:
\begin{align}
& (1): \quad U_1 = I_s \otimes |0\rangle\langle0|_c + B_s \otimes |1\rangle\langle1|_c,\nonumber\\
& (2): \quad U_2 = e^{-itH_s} \otimes I_c \nonumber\\
& (3): \quad U_3 = A_s \otimes I_c, \nonumber\\
& (4): \quad U_4 = e^{itH_s} \otimes I_c , \nonumber\\
& (5): \quad U_5 = B_s \otimes |0\rangle\langle0|_c +I_s \otimes  |1\rangle\langle1|_c,\nonumber
\end{align}
where $A_s$ and $B_s$ are both local unitary operators in the system. Finally, measurement of $\sigma^x_c$ is performed to get the real part of OTOC. A nice property of this protocol is that it works for both pure states and mixed states, which allows straightforward generalization to open systems. 

Note that the above protocol involves both forward and backward time evolution. With dissipation, we assume that only the Hamiltonian of the system is reversed during the backward time evolution while the effect of the environment is unchanged. That is, if forward time evolution is governed by $H_f= H_s+H_e+H_{int}$, then backward time evolution is governed by $H_b= -H_s+H_e+H_{int}$. Correspondingly, the backward dynamical map $\mathcal{V}_b$ and adjoint dynamical map $\mathcal{V}^{\dagger}_b$ differ from the forward ones $\mathcal{V}_f=\mathcal{V}$, $\mathcal{V}^{\dagger}_f=\mathcal{V}^{\dagger}$ by a minus sign in front of $H_s$.

In the presence of dissipation, the full protocol now proceeds as follows. Initially the system is prepared with density matrix $\rho_s(0)$. In addition, a control qubit $c$ is initialized in the state $|+\rangle_c=\frac{1}{\sqrt{2}}(|0\rangle_c+|1\rangle_c)$. The total initial state is $\rho_{\text{init}}=\rho_s(0) \otimes |+\rangle\langle+|_c$. The final state is $\rho_{f}$  after sequentially applying the following super-operators 
\begin{align}
& (1): \quad \mathcal{S}_1 = \mathcal{C} (I_s \otimes |0\rangle\langle0|_c + B_s \otimes |1\rangle\langle1|_c),\nonumber\\
& (2): \quad \mathcal{S}_2 = \mathcal{V}_f(t)\otimes \mathcal{I}_c, \nonumber\\
& (3): \quad \mathcal{S}_3 = \mathcal{C}(A_s\otimes I_c), \nonumber\\
& (4): \quad \mathcal{S}_4 = \mathcal{V}_b(t) \otimes \mathcal{I}_c, \nonumber\\
& (5): \quad \mathcal{S}_5 = \mathcal{C}(B_s\otimes |0\rangle\langle0|_c +I_s \otimes |1\rangle\langle1|_c),\nonumber\\
& \rho_f = \mathcal{S}_5 \cdot \mathcal{S}_4 \cdot \mathcal{S}_3 \cdot \mathcal{S}_2 \cdot \mathcal{S}_1 \cdot \rho_{init},
\end{align}
where $\mathcal{I}$ is the identity super-operator, and the conjugation super-operator is defined by $\mathcal{C}(U)\cdot \rho = U\rho U^\dag$. Finally we perform the measurement $\sigma^x_c$ to get the real part of OTOC
\begin{align}
&F(t, A, B) :=\tr(\sigma^x_c \rho_f)\nonumber \\
&=\Re\,\tr\bigg( \Big(\mathcal{V}^\dag_b(t) \cdot B_s^\dag\Big) A_s \Big(\,\mathcal{V}_f(t) \cdot \big(B_s\rho_s(0)\big) \,\Big) A_s^\dag \bigg),
\end{align}
In this paper, we focus on the case where the initial state of the system is prepared in the equilibrium state at infinite temperature, i.e. $\rho_s(0) = I_s/2^N$ and the unitary operators $A_s$ and $B_s$ are selected as local Pauli operators, for example, $B_s=\sigma^z_1, A_s= \sigma^z_i$. 

\section{Dissipative OTOC corrected for overall decay}\label{sec3}
In this section, we observe that the information light cone disappears due to the fast overall decay of OTOC in dissipative systems. In order to recover the light cone as much as possible, we propose a corrected OTOC to remove the effect of overall decay due to the information leaking in dissipative systems.

In a quantum system without dissipation, the OTOC $F(t,A,B)=\Re\langle A^\dag B_b^\dag(t) A B_b(t)\rangle_{\beta=0}$ has the same capability to reveal the light cones with different time scaling as the operator norm of the commutator $[B_b^\dag(t),A^\dag]$ in the Lieb-Robinson bound \cite{Roberts2015, Roberts:2016aa, huang2016}, where $B^\dag_b(t)$ is the operator $e^{itH_b}B^\dag e^{-itH_b}=e^{-itH_s}B^\dag e^{itH_s}$ in the Heisenberg picture.  When $t < d_{BA}/v_B$, the support of $B^\dag_b(t)$ and $A^\dag$ are approximately disjoint, so $F(t, A, B)$ is almost equal to 1, where $d_{BA}$ is the distance between the local operators $A$ and $B$ and $v_B$ is the butterfly velocity. The OTOC begins to decay \cite{Roberts2015,Roberts:2014ab,Roberts:2016aa,Maldacena2016,Hosur2016} when the support of $B_b^\dag(t)$ grows to $A^\dag$. Furthermore, in chaotic systems, OTOC decays to zero at late time in the thermodynamic limit \cite{Roberts:2014ab,Hosur2016,roberts2016chaos,huang2017}. As shown in the upper left panel of FIG. (\ref{otoczz1}), the OTOC $F(t,A,B)$ is able to reveal the ballistic light cone of information scrambling.

\begin{figure}[h]
\includegraphics[width=1.0\linewidth]{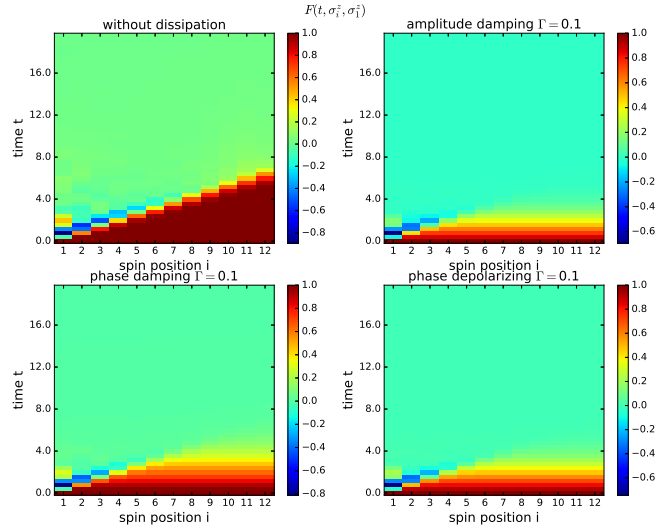}
\caption{OTOC in the chaotic Ising chain (\ref{ising}).}
\label{otoczz1}
\end{figure}

In the presence of dissipation, information is leaking into the environment while being scrambled. Thus $\mathcal{V}^\dag_b(t) \cdot B^\dag$ and the OTOC begins to decay when $t > 0$. Intuitively, dissipation destroys the light cone revealed by the OTOC $F(t,A,B)$ because the OTOC to decay to zero in a short time which is independent of the spatial distance between local operators $A$ and $B$. In FIG. (\ref{otoczz1}), our numerical calculations confirm that the light cone is destroyed. The OTOC $F(t, \sigma^z_i, \sigma^z_1)$ decays to zero for all $i$ approximately when $t>4$. 

In dissipative systems, there are two factors leading to the decay of $F(t, A, B)$: (i) the decay of $\mathcal{V}^\dag_b(t) \cdot B^\dag$ related to the information leaking caused by dissipation, (ii) the non-commutativity between $\mathcal{V}^\dag_b(t) \cdot B^\dag$ and $A^\dag$. Information scrambling is manifested only in (ii) but it might be overshadowed by (i). Is it possible to remove the effect of information leaking and recover the destroyed light cone? One natural idea is to divide the OTOC  $F(t, A, B)$ by a factor representing the decay related to information leaking. The identity operator $I$ commutes with arbitrary operator, and therefore $F(t,I,B)$ is a factor representing the overall decay of quantum information due to leaking only. Therefore, we propose a corrected OTOC to detect the light cone
\begin{align}
\frac{F(t,A, B)}{F(t,I,B)}. 
\label{corrected_oto}
\end{align}

\begin{figure}[h]
\includegraphics[width=1.0\linewidth]{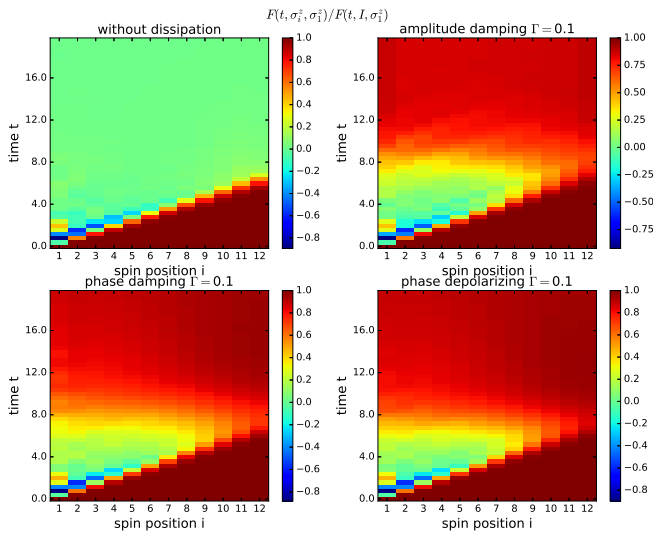}
\caption{Corrected OTOC in the chaotic Ising chain (\ref{ising}).}
\label{otoczz2}
\end{figure}

The numerical results in FIG. (\ref{otoczz2}) show that the corrected OTOC is able to recover the information light cone to some extent in small systems $(N=12)$, with either the dissipation of amplitude damping, phase damping or phase depolarizing.

For small dissipation rate, does the corrected OTOC have the capability to recover the destroyed light cone in the thermodynamic limit? The answer is no. Due to the limited computational resources, we simulate a relatively large system with $24$ spins. FIG. (\ref{light_cone}) shows that the boundary of the light cone revealed by the corrected OTOC gradually disappears in space. Based on this result, we expect that the corrected OTOC only has a finite extent in the thermodynamic limit.

Here let us briefly talk about the numerical methods we used. When $N = 12$, quantum toolbox in Python \cite{qutip1, qutip2} is used to numerically solve the master and adjoint master differential equations (Eqs. (\ref{master eq})(\ref{adjoint master eq})). When $N = 24$, our numerical simulations are based on the time-evolving block decimation (TEBD) algorithm after mapping matrix product operators to matrix product states \cite{vidal2004efficient,zwolak2004mixed,schollwock2011}, which is able to efficiently simulate the evolution of operators or mixed states. In the singular value decomposition, we ignore the singular values $s_k$ if $s_k/s_1<10^{-8}$, where $s_1$ is the maximal one. And the bond dimension is enforced as $\chi\leq 500$. Due to the presence of dissipation, the entanglement growth in the matrix product operator is bounded. Therefore, the OTOC can be efficiently calculated using the TEBD algorithm.

\begin{figure}[h]
\includegraphics[width=0.95\linewidth]{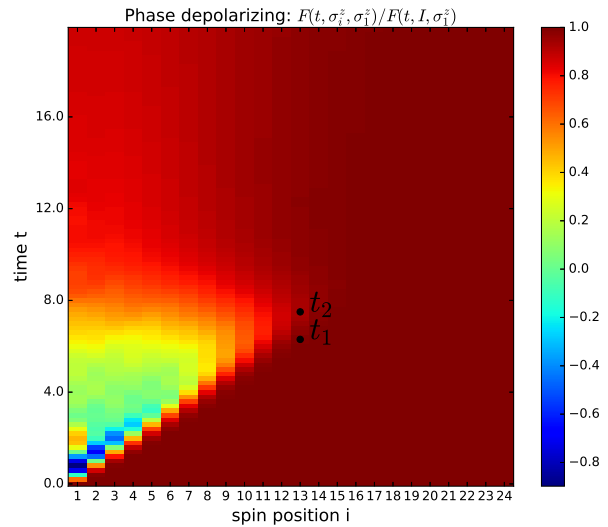}
\caption{Corrected OTOC recovers some part of the original light cone when $N=24$ and $\Gamma=0.1$.}
\label{light_cone}
\end{figure}

\section{The width of the partially recovered light cone}\label{sec4}
The finite extent of the light cone revealed by the corrected OTOC indicates that, besides the overall decay of quantum information, dissipation also leads to structural changes in the scrambled information. In this section, we are going to give a qualitative argument as to why and how the structural change happens.

In particular, we find that the re-structuring happens at late time in two aspects: (i) few-body terms dominate when compared with many-body terms (ii) at fixed time, the weight of few-body terms decays in space. 

Let us define the few-body and many-body terms, and their weights. Consider the operator $B^\dag_b(t) = \mathcal{V}^\dag_b(t) \cdot B^\dag$ which can be written in the basis of products of Pauli matrices as
\begin{align}
B^\dag_b(t) = \sum_S b_S(t) S = \sum_{i_1i_2\cdots i_N} b_{i_1i_2\cdots}(t) \sigma^{i_1}_1\sigma^{i_2}_2\cdots\sigma^{i_N}_N,
\end{align}
where the Pauli string $S$ is a product of Pauli matrices $\sigma^{i_1}_1\sigma^{i_2}_2\cdots\sigma^{i_N}_N$ with $i_k=0,x,y,$ or $z$. In the above decomposition, a few-body (many-body) term is a Pauli string with few (many) non-trivial Pauli matrices. $|b_S(t)|^2 /\sum_{S'} |b_{S'}(t)|^2$ represents the weight of Pauli string $S$.

Our qualitative arguments are mainly based on the Suzuki-Trotter expansion of the adjoint propagator in the infinitesimal time steps
\begin{equation} \label{expansion}
B^\dag_b(t+\tau)=\mathcal{V}^\dag_b(\tau) \cdot B^\dag_b(t) \approx e^{\mathcal{L}^\dag_D\tau}\cdot(B^\dag_b(t) - i\tau[H_s, B^\dag_b(t)]),
\end{equation}
where $\mathcal{L}^\dag_D$ is the adjoint super-operator of the dissipation and $\tau$ is the infinitesimal time interval. Based on this expression, we are able to qualitatively discuss the operator spreading in the space of operators during the time evolution.

The nearest-neighbor interactions in $H_s$ lead to operator growth in space. If there is no dissipation,  every term inside the light cone is expected to have approximately equal weight at late time \cite{huang2017}, so $F(t,A,B)$ is approximately equal to 0 inside the light cone.

Intuitively dissipation leads to operator decay. Many-body terms decay at a higher rate than few-body terms, so few-body terms dominate at late time in dissipative systems. In the channel of phase depolarization, $\mathcal{L}^\dag_D\cdot \sigma^{i_k}=e^{-\Gamma\tau} \sigma^{i_k} (i_k=x,y,z)$. In one step of evolution, the decaying factors of one-body, two-body and $m$-body terms are respectively $e^{-\Gamma \tau}$, $e^{-2\Gamma \tau}$ and $e^{-m\Gamma \tau}$. Many-body terms decay faster than few-body terms. Amplitude and phase damping channels have similar behaviors. In the dominating few-body terms, firstly we need to consider one-body terms. Secondly, the nearest-neighbor two-body terms cannot be ignored because the nearest-neighbor interactions in $H_s$ (Eq. (\ref{expansion})) transform one-body operators into nearest-neighbor two-body operators. Our simulations support these qualitative arguments. FIG. (\ref{weights}) shows that the sum of the weights of one-body and nearest-neighbor two-body terms approximately exceeds 90\% at late time in the dissipative channels. 

\begin{figure}[h]
\includegraphics[width=0.9\linewidth]{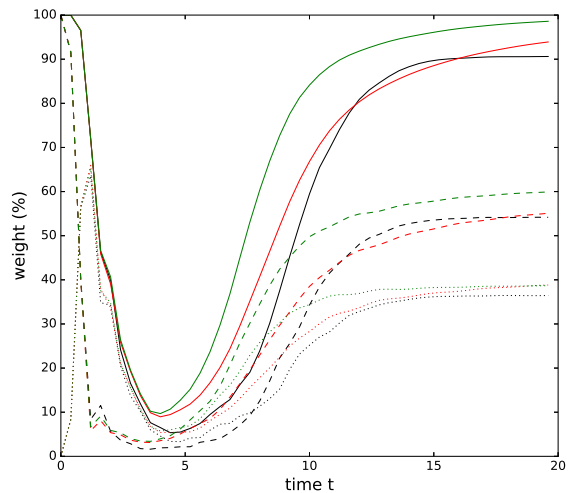}
\caption{Color plot of the weights of few-body terms in the chaotic Ising chain (\ref{ising}) with $N=12$ spins and dissipation rate $\Gamma=0.1$. Dashed (dotted) line denotes the total weight of one-body (nearest-neighbor two-body) terms in the operator $\mathcal{V}^\dag_b(t) \cdot\sigma^z_1$, while the solid line is the sum of dotted and dashed lines. Black, red, and green lines are the results for dissipative channels of amplitude damping, phase damping and phase depolarizing respectively.}
\label{weights}
\end{figure}

Moreover, because of dissipation, the weight of few-body terms decays in space at the same time. In the time-evolving operator $\mathcal{V}^\dag_b(t) \cdot\sigma^z_1$, few-body terms to the right are sequentially generated from the ones to the left. For example, one-body term $\sigma^{i_{k+1}}_{k+1}$ is generated via the path $\sigma^{i_k}_k \rightarrow \sigma^{i'_{k}}_k\sigma^{i'_{k+1}}_{k+1}\rightarrow \sigma^{i_{k+1}}_{k+1}$, where $i_k, i'_k, i'_{k+1}, i_{k+1}$ are non-trivial indicies $x, y$ or $z$. Considering the generating paths and the different decaying rate of few-body terms, we find that extra spatial decaying factor exists when comparing the coefficients of $\sigma^{i_{k+1}}_{k+1}$ and $\sigma^{i_k}_k$. Spacial decaying factors accumulate during the scrambling of information, so the weight of few-body terms decays in space at the same time. Thus $(1 - F(t,\sigma^z_k,\sigma^z_1)/F(t,I,\sigma^z_1))$, which is proportional to the coefficient $b_S(t)$ of few-body terms $S$ near site $k$, gradually vanishes in space. FIG. (\ref{light_cone}) confirms this point.

Besides the qualitative discussions, we are going to quantitatively study the relationship between the width $d(\Gamma)$ of the partially recovered light cone and the dissipation rate $\Gamma$. Appendix A provides a lower bound $\sqrt{\epsilon a v_{LR}/\Gamma}$, where $a$ is the distance between two nearest neighbor sites, $v_{LR}$ is the Lieb-Robinson velocity and $\epsilon$ is a small number. This inequality is shown to be satisfied for the width of the light cone revealed by the corrected OTOC in the channel of phase damping or phase depolarizing. In general, we expect that $d(\Gamma)$ obeys a power law $c/\Gamma^\alpha$ when the dissipation rate $\Gamma$ is sufficiently small.

Now we discuss how to find the width $d(\Gamma)$ of the partially recovered light cone in the numerical calculations. Our criterion is that if the difference of corrected OTOCs at $(t_1=(d_{BA}-w/2)/v_B, d_{BA})$ and $(t_2=(d_{BA}+w/2)/v_B,d_{BA})$ (see FIG. \ref{light_cone}) is less than a threshhold value $\delta$, for example $0.1$, then it is impossible to recognize the boundary of the light cone and we identify the smallest such $d_{BA}$ as the width of the recovered light cone. Here $w$ is the width of the boundary of the light cone in the system without dissipation and $v_B$ is the corresponding butterfly velocity.

\begin{figure}[h]
\includegraphics[width=0.90\linewidth]{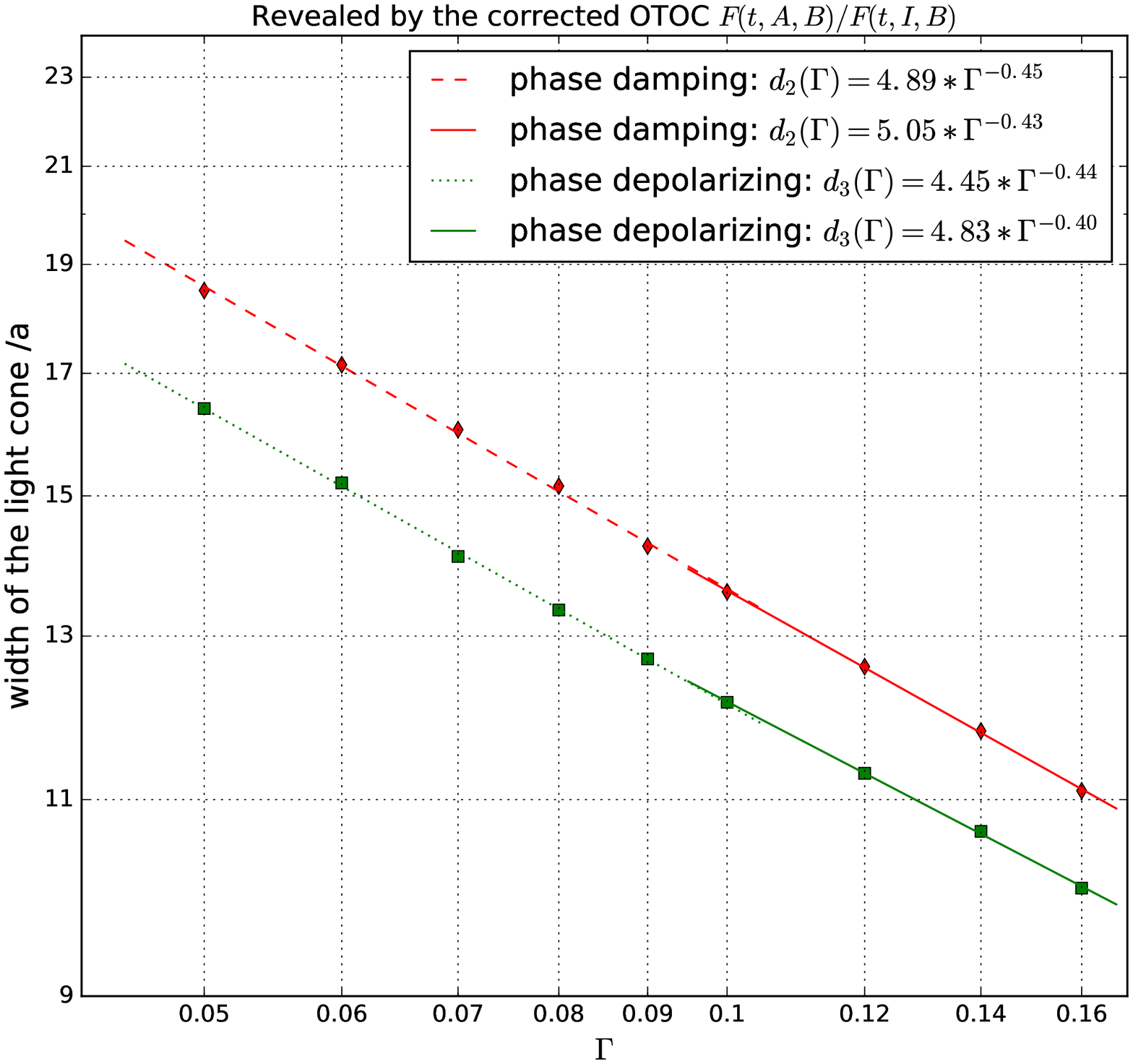}
\caption{The log-log plot of $d(\Gamma)$ and $\Gamma$.}
\label{power_decay}
\end{figure}

Our numerical simulation supports that $d(\Gamma)$ obeys a power law $c/\Gamma^\alpha$. In FIG. (\ref{power_decay}), our fitting results are: $\alpha_2\approx0.45, \alpha_3\approx0.44$ when $0.05\leq \Gamma\leq 0.1$, and $\alpha_2\approx0.43, \alpha_3\approx0.40$ when $0.1\leq \Gamma\leq 0.16$, where the subscripts $2,3$ represent the channel of phase damping and phase depolarizing respectively. If $\Gamma$ is sufficiently small, the power-law value $c /\Gamma^\alpha$ is expected to be greater than or equal to the lower bound $\sqrt{\epsilon a v_{LR}/\Gamma}$. This implies that $\alpha$ should be greater than or equal to $0.5$. Here in our simulation, $\alpha_2$ and $\alpha_3$ are smaller than $0.5$. The reason is that the dissipation rates in the range of $[0.05, 0.1]$ are not small enough. Theoretically, the derivations in Appendix A give the condition of sufficiently small $\Gamma$ via comparing $\sqrt{\epsilon a v_{LR}/\Gamma}$ with $\xi$. $\Gamma$ is sufficiently small if it is approximately less than $\epsilon a v_{LR}/(9\xi^2)$. In this chaotic Ising model, after  selecting $\epsilon \sim 0.1$, and estimating the parameters $v_{LR} \sim 1.7 a, \xi \sim a $, then we obtain that $\Gamma \lesssim  0.01 $ is sufficiently small. Therefore, our numerical result does not contradict the lower bound proved in Appendix A. Numerically, we see that $\alpha$ decreases when the range of $\Gamma$ increases. 

Even though amplitude damping has different properties when compared with phase damping and phase depolarizing, we numerically verify that $d(\Gamma)$ still scales a power law of the dissipation rate $\Gamma$. In the channel of amplitude damping, the corrected OTOC depends on $\mathcal{V}^\dag_b(t)$ and $\mathcal{V}_f(t)$ which have different properties. The identity is a fixed point of  $\mathcal{V}^\dag_b(t)$ while $\mathcal{V}_f(t)$ is trace-preserving. The proof in Appendix A does not apply to amplitude damping, thus the lower bound $\sqrt{\frac{\epsilon av_{LR}}{\Gamma}}$ does not work for the corrected OTOC in this channel. In the numerical simulation, we confirm that the general expectation of power-law decay is still correct. FIG. (\ref{power_decay1}) shows that $d(\Gamma)$ scales as a power law of $\Gamma$ with the power $\alpha_1\approx0.31$ when $0.05\leq \Gamma\leq 0.1$, where the subscript $1$ represents the channel of amplitude damping.

\begin{figure}[h]
\includegraphics[width=0.90\linewidth]{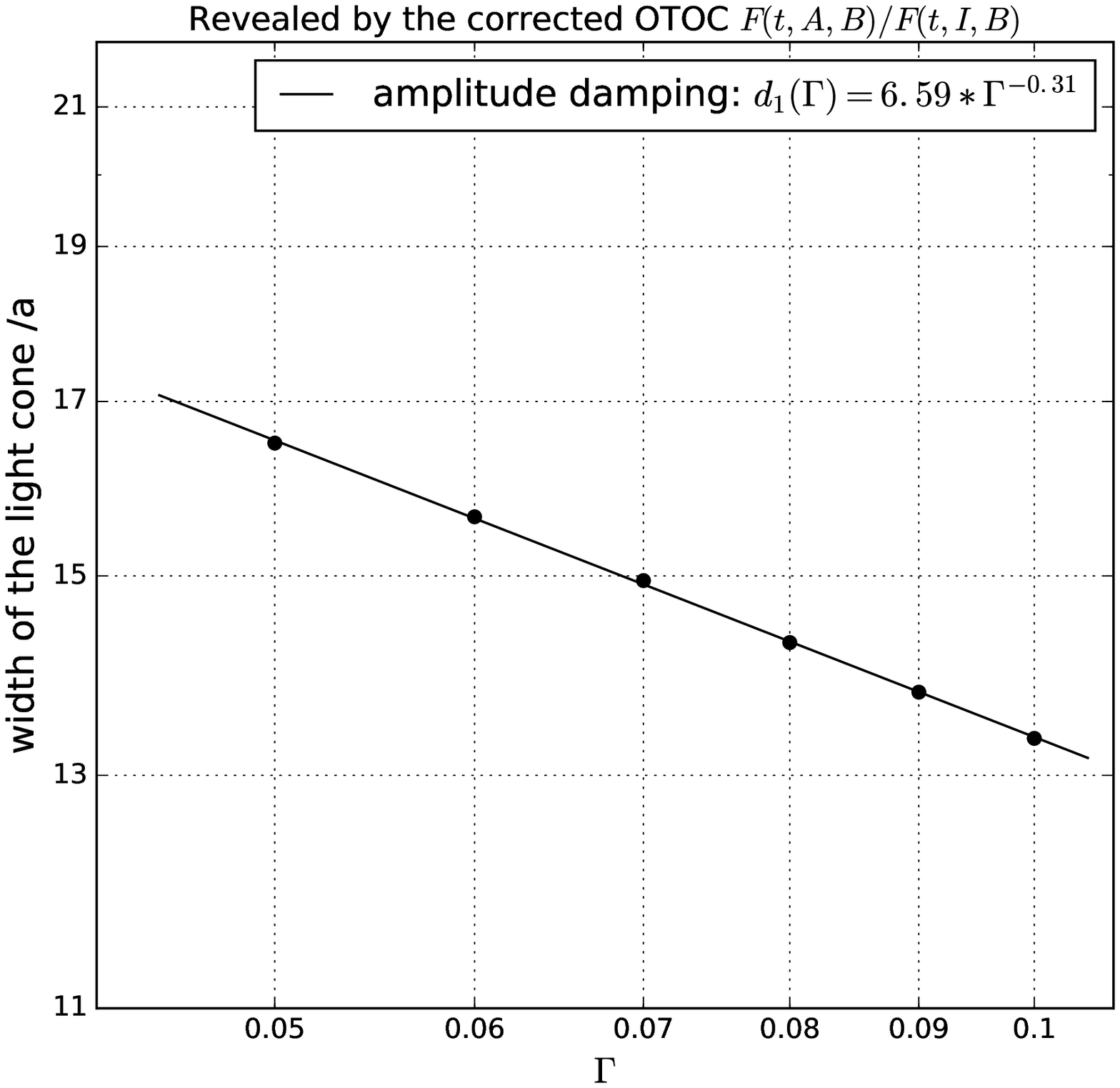}
\caption{The log-log plot of $d(\Gamma)$ and $\Gamma$.}
\label{power_decay1}
\end{figure}

\section{Lieb-Robinson bound in dissipative systems}\label{sec5}
Now we would like to discuss the Lieb-Robinson bound and its connections with OTOC in open quantum systems. Based on the observation of corrected OTOC, we conjecture a tighter Lieb-Robinson bound for dissipative systems.

The Lieb-Robinson inequality provides an upper bound for the speed of information propagation in quantum systems with local interactions. Let us briefly review the Lieb-Robinson bound.

Two observers, Alice and Bob, have access to the quantum system. The system is initially in the state $\rho(0)$ and its dynamics is governed by the dynamical map $\mathcal{V}_b(t)$ related to the Hamiltonian $H_b =- H_s + H_e + H_{int}$. The sender Alice has the option to perform some local actions in her region. After some time $t$, the receiver Bob performs some measurements to detect the signal. No signal is sent to Bob if Alice does nothing. In order to send a signal, Alice performs a small local unitary perturbation $U_A=e^{-i\epsilon O_A}$ in her region, which maps the state $\rho_s(0)$ to $\rho'_s(0)=U_A \rho_s(0) U_A^\dag\approx \rho_s(0)-i\epsilon [O_A, \rho_s(0)]$, where $O_A$ is a local Hermitian operator. At time $t$, Bob makes a measurement described by the local Hermitian operator $O_B$. The difference of outcomes describing the capability to detect the signal is
\begin{align}
&\left|\tr\Big(O_B \mathcal{V}_b(t) \cdot \big(\rho'_s(0)-\rho_s(0)\big)\Big)\right| \nonumber \\
& = \epsilon |\tr(\rho_s(0) [\mathcal{V}_b^\dag(t) \cdot O_B, O_A] )|\nonumber \\
&\leq\epsilon \| \,[\mathcal{V}_b^\dag(t) \cdot O_B, O_A] \,\|, 
\end{align}
where the operator norm is defined by $\|O\| = \sup_{|\psi\rangle} \| O|\psi\rangle\| / \||\psi\rangle\|$. Following the Lieb-Robinson bound in closed systems \cite{LR72,HK06,NOS06propg}, an inequality has been proved in open quantum systems \cite{poulin2010lieb,NVZ2011lieb,BK2012,descamps2013,kliesch2014lieb}
\begin{align}
\|\, [\mathcal{V}^\dag_b(t)\cdot O_B, O_A] \,\| \leq c \, \|O_A\| \cdot \|O_B\| \; e^{-\frac{d_{BA} -v_{LR}t}{\xi}},
\label{lieb_bound}
\end{align}
where $c, \xi$ are some constants, $v_{LR}$ is the Lieb-Robinson velocity, and $d_{BA}$ is the distance between the local operators $O_A$ and $O_B$. The Lieb-Robinson velocity $v_{LR}$ is an upper bound for the speed of information propagation, so it is greater than or equal to the butterfly velocity $v_B$ at $\beta=0$ in Eq. (\ref{butterfly_velocity}) \cite{Roberts2015}. Refs. \cite{Roberts:2016aa,Hosur2016,Mezei2017,Lucas:2017aa} provide more discussions about the relationship between $v_B$ and $v_{LR}$. 

In dissipative systems, the left-hand side of Eq. (\ref{lieb_bound}) decays to zero at late time, so Eq. (\ref{lieb_bound}) is not tight enough. One reason is that the operator $\mathcal{V}_b^\dag(t)\cdot O_B$ in the Heisenberg picture is overall decaying because of the dissipation. Ref. \cite{BK2012} has proved that the operator norm of $\mathcal{V}_b^\dag(t)\cdot O_B$ is non-increasing because of the dissipation, i.e. $ \| \mathcal{V}_b^\dag(t+\mathrm{d}t)\cdot O_B \| \leq \| \mathcal{V}_b^\dag(t)\cdot O_B \| $, where $\mathrm{d}t$ is an infinitesimal time step. This means that the non-trivial elements in the time-evolving operator are decaying during the time evolution. Our numerical simulations (FIG. \ref{liebzz1}) show that the left-hand side of Eq. (\ref{lieb_bound}) decays to zero at late time, and the boundary of the light cone gradually disappears when the distance $d_{BA}$ increases.

\begin{figure}[h]
\includegraphics[width=1.0\linewidth]{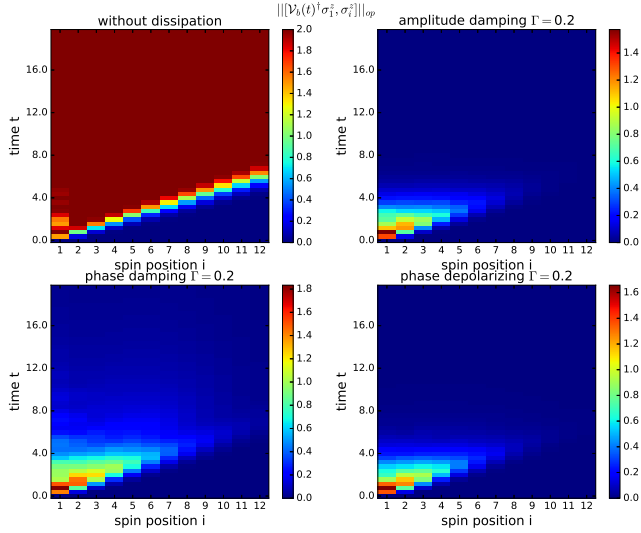}
\caption{The operator norm of the commutator in the chaotic Ising chain (\ref{ising}).}
\label{liebzz1}
\end{figure}

Inspired by the corrected OTOC, we conjecture a tighter Lieb-Robinson bound in dissipative systems
\begin{align}
\quad \frac{\|\, [\mathcal{V}^\dag_b(t)\cdot O_B, O_A] \,\|} { \|O_A\| \cdot \|\mathcal{V}^\dag_b(t)\cdot  O_B\|}  \leq c \, e^{-\frac{d_{AB} -v_{LR}t}{\xi}}.
\label{lieb_bound2}
\end{align}

The above tighter bound has deep connections with the corrected OTOC. In the channel of phase damping or phase depolarizing, the adjoint dynamical map $\mathcal{V}^\dag_b(t)$ is exactly equal to $\mathcal{V}_f(t)$, then $2\Big(1 - \frac{F(t,A, B)}{F(t,I,B)}\Big)=\frac{\|\, [\mathcal{V}^\dag_b(t)\cdot O_B, O_A] \,\|^2_{F}} { \|\mathcal{V}^\dag_b(t)\cdot  O_B\|^2_{F}}$, where $\|O\|_{F}=\sqrt{\tr(OO^\dag)/2^N}$ is the normalized Frobenius norm of the operator $O$. We expect that the normalized Frobenius and operator norm exhibit similar behaviors during the time evolution. Based on this expectation, Eq. (\ref{lieb_bound2}) is conjectured in dissipative systems via changing the normalized Frobenius norm to the operator norm. Similar to the corrected OTOC, the left-hand side of the above modified version of Lieb-Robinson bound is able to partially recover the destroyed light cone in the chaotic Ising chain with dissipation (see FIG. \ref{liebzz2}).

\begin{figure}[h]
\includegraphics[width=1.0\linewidth]{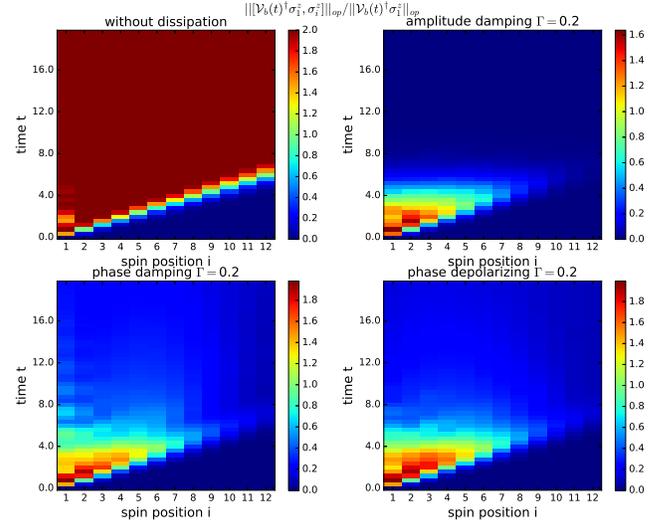}
\caption{The corrected operator norm of the commutator in the chaotic Ising chain (\ref{ising}).}
\label{liebzz2}
\end{figure}

In the above tighter Lieb-Robinson bound, the correcting factor $1 / \|\mathcal{V}^\dag_b(t)\cdot  O_B\|$ has different behaviors in different dissipative channels. $\|\mathcal{V}^\dag_b(t)\cdot  O_B\|$ decays to zero in the channel of phase damping or phase depolarizing but converges to a positive constant in the channel of amplitude damping (see FIG. \ref{opnorm_decay}). In the channel of amplitude damping, the adjoint dynamical map $\mathcal{V}^\dag_b(t)$ does not preserve the trace of an operator, the identity operator $I$ appears in the decomposition of $\mathcal{V}^\dag_b(t)\cdot O_B$ in terms of Pauli operators when $O_B$ is traceless. Therefore, the operator norm of $\mathcal{V}^\dag_b(t)\cdot O_B$ converges to a constant. This can also be observed in the upper right panel of FIG. (\ref{liebzz2}) which is distinct from the lower ones. The operator norm of the commutator is decaying to zero while the denominator converges to a positive constant when $t > 7$. In the channel of amplitude damping, the correcting factor $1/\|\mathcal{V}^\dag_b(t)\cdot  O_B\|$ does not play an essential role to remove the effect of overall decay due to the information leaking.

\begin{figure}
\includegraphics[width=0.9\linewidth]{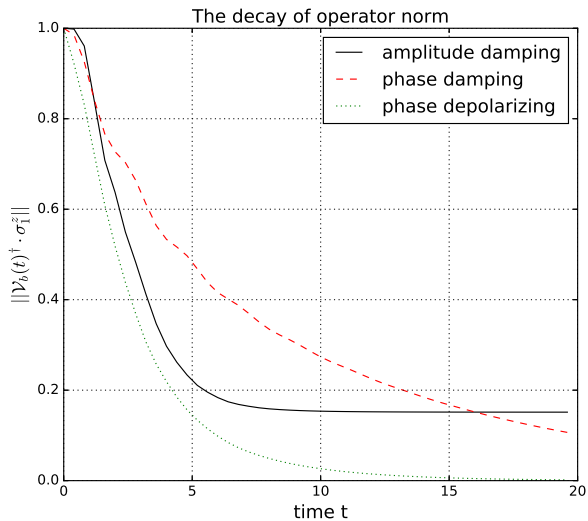}
\caption{The decay of operator norm $\|\mathcal{V}^\dag_b(t)\cdot\sigma^z_1\|$ in different channels ($N=12, \Gamma=0.2$).}
\label{opnorm_decay}
\end{figure}

\section{Conclusions and discussions}
In this paper, we study the effect of dissipation on information scrambling in open chaotic systems. By numerically calculating the measured OTOC signal in a chaotic spin chain in the presence of common types of dissipation, we find that dissipation leads to the decay of the signal not only due to information leaking, but also information re-structuring. We define a corrected OTOC to remove the effect of leaking and partially recover the information light cone. However, due to the re-structuring, the recovered light cone only persists to a finite distance. Based on this understanding of how dissipation affects information scrambling, we conjecture a tighter version of the Lieb-Robinson bound in open systems, which we support with numerical simulation. 

Given the observation we made in this paper, several open questions would be interesting to explore in future work. First, we qualitatively discussed the information re-structuring during scrambling. A more accurate estimation of the size of the light cone may be obtained by carefully modeling the dynamics as dissipative quantum walks. Secondly, although we were able to partially recover the light cone numerically, this is not practical experimentally, as the normalization factor we divide out in Eq. (\ref{corrected_oto}) decays exponentially in time and quickly becomes too small to be accessible experimentally. Is there a better way to see information scrambling in the presence of dissipation? Are there quantities which are also sensitive to information scrambling as OTOC but more robust to the effect of dissipation? This is an important question to be addressed in future work. Finally, we conjectured the modified version of open system Lieb-Robinson bound based on numerical observation. It would be nice to see if this bound can be analytically proved.

\begin{acknowledgments}
When we were finishing this manuscript, we learned of the work by Swingle and Yunger Halpern \cite{Swingle:2018aa} which also studies the problem of extracting OTOCs' early-time dynamics in the presence of error and decoherence.

Y.-L.Z., Y.H., and X.C. are supported by National Science Foundation under award number DMR-1654340 and the Alfred P. Sloan research fellowship. We acknowledge funding provided by the Institute for Quantum Information and Matter, an NSF Physics Frontiers Center (NSF Grant PHY-1125565) with support of the Gordon and Betty Moore Foundation (GBMF-2644). X.C. is also supported by the Walter Burke Institute for Theoretical Physics at Caltech.
\end{acknowledgments}

\appendix
\section*{Appendix A: Proof of a lower bound \label{appendix1: A}}
Here we prove a lower bound $\sqrt{\epsilon a v_{LR}/\Gamma}$ for the width $d(\Gamma)$ of the partially recovered light cone revealed by the corrected OTOC in the channel of phase damping or phase depolarizing. The main ideas in the proof are comparing the difference between the adjoint propagator in the dissipative channel and the unitary one without dissipation, and employing the adjoint propagator of spatially truncated adjoint Liouvillians.

\begin{lemma} \label{lemma1}
Suppose $\mathcal{L}^\dag_1(t)$ and $\mathcal{L}^\dag_0(t)$ are the adjoint Liouvillian super-operators describing Markovian dynamics of the same open quantum system with $\|\mathcal{L}^\dag_1(t) - \mathcal{L}^\dag_0(t) \|\leq f(t)$, then the difference of adjoint propagators satisfies $\|\mathcal{V}^\dag_1(t,0) -\mathcal{V}^\dag_0(t,0)\|\leq \int_0^t d\tau f(\tau)$,  where $\mathcal{V}^\dag_k(t,s) = \mathcal{T}_{\rightarrow} e^{\int_s^t \mathcal{L}^\dag_k(\tau)d\tau} (t \geq s, k=0,1)$, and $\mathcal{T}_{\rightarrow}$ or $\mathcal{T}_{\leftarrow}$ is the time-ordering operator which orders products of time-dependent operators such that their time arguments increase in the direction indicated by the arrow.
\end{lemma}

\begin{proof}
\begin{align}
& \|\mathcal{V}^\dag_1(t,0) -\mathcal{V}^\dag_0(t,0)\| = \|\mathcal{V}^\dag_0(0,0) \mathcal{V}^\dag_1(t,0) - \mathcal{V}^\dag_0(t,0) \mathcal{V}^\dag_1(t,t)  \| \nonumber\\
& =\|\int_0^t ds \frac{\partial}{\partial s}\Big(\mathcal{V}^\dag_0(s,0) \mathcal{V}^\dag_1(t,s)\Big)\|  \nonumber\\
& \leq \int_0^t ds \| \mathcal{V}^\dag_0(s,0) (\mathcal{L}^\dag_0(s)-\mathcal{L}^\dag_1(s))\mathcal{V}^\dag_1(t,s)\|  \nonumber\\
& \leq \int_0^t ds \| \mathcal{V}^\dag_0(s,0)\| \cdot \|\mathcal{L}^\dag_1(s)-\mathcal{L}^\dag_0(s)\|\cdot \| \mathcal{V}^\dag_1(t,s)\|  \nonumber\\
& \leq \int_0^t ds \|\mathcal{L}^\dag_1(s)-\mathcal{L}^\dag_0(s)\|\leq \int_0^t ds f(s)\nonumber
\end{align} 
In the derivation, one uses the fact that the adjoint propagators $\mathcal{V}^\dag_k(t,s)$ are norm-nonincreasing \cite{breuer2002,poulin2010lieb,BK2012}.
\end{proof}

Here, we need to pay attentions to the difference between the propagator $\mathcal{V}(t,s) = \mathcal{T}_{\leftarrow} e^{\int_s^t \mathcal{L}(\tau)d\tau} (t \geq s)$  and its adjoint $\mathcal{V}^\dag(t,s) = \mathcal{T}_{\rightarrow} e^{\int_s^t \mathcal{L}^\dag(\tau)d\tau} (t \geq s)$. $\mathcal{V}(t,s)$ is acting on the density matrix and trace-preserving. $\mathcal{V}^\dag(t,s)$ is acting on the observables and the identity is one of its fixed points. For unitary evolution, $\mathcal{V}^\dag(t,s)$ and $\mathcal{V}(t,s)$ are the inverse of each other and both norm-preserving. When dissipation exists, only $\mathcal{V}^\dag(t,s)$ is norm-nonincreasing for arbitrary observables, i.e. $\|\mathcal{V}^\dag(t,s) \cdot O \|\leq \|O\|$ ($\forall O = O^\dag$).

\begin{lemma}\label{lemma2}
In a one-dimensional system, $\mathcal{L}^\dag_H=\sum_i \mathcal{L}^\dag_{H_i} $ is the sum of local adjoint Liouvillian super-operators, and $\mathcal{L}^\dag_D=\Gamma \sum_k \mathcal{L}^\dag_{D,k}$ is the sum of adjoint dissipative super-operators acting on each site, where $\|\mathcal{L}^\dag_{D,k}\|\leq 1$ and $\Gamma$ is the dissipation rate. During the evolution, the operator difference between the dissipative and unitary channel is bounded by $\|\mathcal{V}^\dag_1(t,0) \cdot B -\mathcal{V}^\dag_0(t,0)\cdot B\| \leq \Gamma O( t^2)$, where $\mathcal{V}^\dag_1(t,s) = \mathcal{T}_{\rightarrow} e^{\int_s^t (\mathcal{L}^\dag_H(\tau)+\mathcal{L}^\dag_D(\tau)) d\tau} (t \geq s)$, $\mathcal{V}^\dag_0(t,s) = \mathcal{T}_{\rightarrow} e^{\int_s^t \mathcal{L}^\dag_H(\tau) d\tau} (t \geq s)$, and $B$ is a local observable at site 0.
\end{lemma}

\begin{proof}
For an open quantum system described by short-range Liouvillians, the Lieb-Robinson bound
\begin{equation}
\|\, [\mathcal{V}^\dag(t)\cdot B, A] \,\| \leq c \, \|B\| \cdot \|A\| \; e^{-(d_{AB} -v_{LR}t)/\xi}
\end{equation}
implies the existence of an upper limit to the speed of quantum information propagation. The outside signal is exponentially small with the distance from the boundary of the effective light cone. Based on the Lieb-Robinson bound, Ref. \cite{BK2012} obtained the quasi-locality of Makovian quantum dynamics: up to exponentially small error, the evolution of local observable can be approximately obtained by applying the propagator of a spatially truncated version of the adjoint Liouvillian, provided that the range of the truncated propagator is larger than the support of the time-evolving observable. The truncated propagators we select are
\begin{align}
&\tilde{\mathcal{V}}^\dag_0(t,0)=\mathcal{T}_{\rightarrow} e^{\int_0^t d\tau \sum_{i: H_i \subset \Lambda(\tau)} \mathcal{L}^\dag_{H_i} (\tau) }, \ \tilde{\mathcal{V}}^\dag_1(t,0)= \nonumber\\
& \mathcal{T}_{\rightarrow} e^{\int_0^t d\tau \big(\sum_{i: H_i \subset\Lambda(\tau)}\mathcal{L}^\dag_{H_i}(\tau) +  \sum_{k=- d(\tau)/a}^{d(\tau)/a} \mathcal{L}_{D,k}(\tau) \big)},\nonumber
\end{align}
where $d(t)\sim 2v_{LR}t + \xi$, $a$ is the distance between two nearest neighboring sites, and  $H_i \subset \Lambda(t)$ means the local term $H_i$ is located in the regime $\Lambda(t)=(-d(t),d(t))$. Let $B_k(t)=\mathcal{V}^\dag_k(t,0) \cdot B$ and $\tilde{B}_k(t)=\tilde{\mathcal{V}}^\dag_k(t,0) \cdot B $, by applying the triangle inequality, one obtains
\begin{align} \label{tri_ineq}
\|B_1(t) -B_0(t)\| \leq & \|\tilde{B}_1(t)- \tilde{B}_0(t)\| + \nonumber \\
& \|B_1(t) - \tilde{B}_1(t)\| + \|\tilde{B}_0(t)-B_0(t)\|. \nonumber
\end{align}
For the right-hand site, the first quantity is bounded by $\| B\|\int_0^t d\tau \Gamma (\xi/a+2v_{LR}\tau/a) =  \| B\| \Gamma t (v_{LR}t+\xi)/a$ (\textbf{Lemma 1}), the second and third one both are less than or equal to $c'\; \|B\|\; e^{(v_{LR}t-d(t))/\xi}= c'\; \|B\|\; e^{-1-v_{LR}t/\xi}$ \cite{BK2012}. When $t\gtrsim 3\xi/v_{LR}$, then $\|\mathcal{V}^\dag_1(t,0) \cdot B -\mathcal{V}^\dag_0(t,0)\cdot B\|\sim v_{LR}\Gamma t^2/a$. Therefore $\|\mathcal{V}^\dag_1(t,0) \cdot B -\mathcal{V}^\dag_0(t,0)\cdot B\| \leq \Gamma O(t^2)$.
\end{proof}

\begin{prop}\label{myprop1}
In the chaotic Ising chain with dissipations acting on each site, the light cone within the time range $t\leq \sqrt{\frac{\epsilon a}{v_{LR} \Gamma}}$ can be revealed by $\| [\mathcal{V}^\dag_1(t,0)\cdot B, A] \|$, $\|[\mathcal{V}^\dag_1(t,0)\cdot B,A]\|/\|\mathcal{V}^\dag_1(t,0)\cdot B\| $, $\| [\mathcal{V}^\dag_1(t,0)\cdot B, A] \|_F$ and $\|[\mathcal{V}^\dag_1(t,0)\cdot B,A]\|_F/\|\mathcal{V}^\dag_1(t,0)\cdot B\|_F $, where $v_{LR}$ is the Lieb-Robinson velocity, $a$ is the distance between two nearest neighboring sites, $\epsilon$ is a small number (for example, $\epsilon \sim 0.1$), and $\Gamma$ is the sufficiently small dissipation rate $(\ll \epsilon a v_{LR}/\xi^2)$, $\| O \|$ is the operator norm and $\|O\|_{F} =\lim_{N\rightarrow\infty}\sqrt{\tr(OO^\dag)/2^N}$ is the normalized Frobenius norm of operators in the thermodynamic limit. The width of the light cone is at least $\sqrt{\frac{\epsilon av_{LR}}{\Gamma}}$.
\end{prop}

\begin{proof}
According to \textbf{Lemma 2}, if $ 3\xi/v_{LR} \lesssim t\leq \sqrt{\frac{\epsilon a}{v_{LR} \Gamma}}$, then one obtains $\|B_1(t) -B_0(t)\|\leq \epsilon \|B\|$ when comparing the operators $B_1(t)=\mathcal{V}^\dag_1(t,0)\cdot B$ in the dissipative channel and $B_0(t)=\mathcal{V}^\dag_0(t,0)\cdot B$ in the unitary channel. Applying the triangle inequality, one obtains
\begin{align}
& (1-\epsilon) \|B\| \leq \| B_1(t) \| \leq   (1+\epsilon)\|B\|, \nonumber \\
& \|[B_0(t), A]\| - \epsilon \|\mathcal{C}_A\| \| B\| \leq \|[B_1(t),A]\| \leq \nonumber\\
& \|[B_0(t), A]\| + \epsilon \|\mathcal{C}_A\| \| B\|, \nonumber \\
& (1+\epsilon)^{-1}\Big(\frac{\|[B_0(t), A]\|}{\|B\|} - \epsilon \|\mathcal{C}_A\|\Big) \leq \frac{\|[B_1(t),A]\|}{\|B_1(t)\| } \leq \nonumber\\ & (1-\epsilon)^{-1}\Big(\frac{\|[B_0(t), A]\|}{\|B\|} + \epsilon \|\mathcal{C}_A\|\Big),\nonumber
\end{align}
where the super-operator $\mathcal{C}_A$ is defined by $\mathcal{C}_A \cdot O =[O,A]$. The normalized Frobenius norm is less than or equal to the operator norm, i.e. $\|O\|_{F} \leq \|O\|$, so we get 
\begin{align}
&\|B_1(t)-B_0(t) \|_{F} \leq \|B_1(t)-B_0(t) \| \leq \epsilon \|B\| \nonumber\\
& \|[B_0(t), A]\|_F - \epsilon \|\mathcal{C}_A\| \| B\| \leq \|[B_1(t),A]\|_F \leq \nonumber\\
& \|[B_0(t), A]\|_F + \epsilon \|\mathcal{C}_A\| \| B\|, \nonumber \\
& \frac{\|[B_0(t), A]\|_{F} - \epsilon \|\mathcal{C}_A\| \|B\|}{\|B\|_{F} + \epsilon \|B\|} \leq\frac{\|[B_1(t),A]\|_{F}}{\|B_1(t)\|_{F} } \leq \nonumber\\
&\frac{\|[B_0(t), A]\|_{F} + \epsilon \|\mathcal{C}_A\| \|B\|}{\|B\|_{F} - \epsilon \|B\|}.\nonumber
\end{align}

In the unitary channel, $\|[B_0(t), A]\|$ and $\|[B_0(t), A]\|_{F}$ both are able to detect the light cone. Because $\epsilon$ is a small number, it is also small that the difference of the corresponding quantities between the dissipative and unitary channel. Thus, $\| [B_1(t), A] \|$, $\|[B_1(t),A]\|/\|B_1(t)\| $, $\| [B_1(t), A] \|_F$ and $\|[B_1(t),A]\|_F/\|B_1(t)\|_F$ are both able to detect the light cone in the time range  $ 3\xi/v_{LR} \lesssim t\leq \sqrt{\frac{\epsilon a}{v_{LR} \Gamma}}$. The width of the light cone is at least $\sqrt{\frac{\epsilon av_{LR}}{\Gamma}}$ for sufficiently small dissipation rate $\Gamma \ll \epsilon a v_{LR}/\xi^2$.
\end{proof}

\begin{corollary}\label{corollary1}
For sufficiently small dissipation rate $\Gamma \ll \epsilon a v_{LR}/\xi^2$, the lower bound $\sqrt{\frac{\epsilon av_{LR}}{\Gamma}}$ works for the width of the light cone revealed by the corrected OTOC in the chaotic Ising chain with dissipation of phase damping or phase depolarizing.
\end{corollary}

\begin{proof}
In the channel of phase damping or phase depolarizing, the adjoint propagator $\mathcal{V}^\dag_b(t)$ is exactly equal to the propagator $\mathcal{V}_f(t)$, then
\begin{equation}
2\Big(1 - \frac{F(t,A, B)}{F(t,I,B)}\Big)=\frac{\|\, [\mathcal{V}^\dag_b(t)\cdot B, A] \,\|^2_{F}} { \|\mathcal{V}^\dag_b(t)\cdot B\|^2_{F}}.
\end{equation}
Based on \textbf{Proposition 1}, the lower bound $\sqrt{\frac{\epsilon av_{LR}}{\Gamma}}$ works for the width of the light cone revealed by the corrected OTOC in the channel of phase damping or phase depolarizing.
\end{proof}

\bibliographystyle{apsrev4-1}

\begin{thebibliography}{55}%
\makeatletter
\providecommand \@ifxundefined [1]{%
 \@ifx{#1\undefined}
}%
\providecommand \@ifnum [1]{%
 \ifnum #1\expandafter \@firstoftwo
 \else \expandafter \@secondoftwo
 \fi
}%
\providecommand \@ifx [1]{%
 \ifx #1\expandafter \@firstoftwo
 \else \expandafter \@secondoftwo
 \fi
}%
\providecommand \natexlab [1]{#1}%
\providecommand \enquote  [1]{``#1''}%
\providecommand \bibnamefont  [1]{#1}%
\providecommand \bibfnamefont [1]{#1}%
\providecommand \citenamefont [1]{#1}%
\providecommand \href@noop [0]{\@secondoftwo}%
\providecommand \href [0]{\begingroup \@sanitize@url \@href}%
\providecommand \@href[1]{\@@startlink{#1}\@@href}%
\providecommand \@@href[1]{\endgroup#1\@@endlink}%
\providecommand \@sanitize@url [0]{\catcode `\\12\catcode `\$12\catcode
  `\&12\catcode `\#12\catcode `\^12\catcode `\_12\catcode `\%12\relax}%
\providecommand \@@startlink[1]{}%
\providecommand \@@endlink[0]{}%
\providecommand \url  [0]{\begingroup\@sanitize@url \@url }%
\providecommand \@url [1]{\endgroup\@href {#1}{\urlprefix }}%
\providecommand \urlprefix  [0]{URL }%
\providecommand \Eprint [0]{\href }%
\providecommand \doibase [0]{http://dx.doi.org/}%
\providecommand \selectlanguage [0]{\@gobble}%
\providecommand \bibinfo  [0]{\@secondoftwo}%
\providecommand \bibfield  [0]{\@secondoftwo}%
\providecommand \translation [1]{[#1]}%
\providecommand \BibitemOpen [0]{}%
\providecommand \bibitemStop [0]{}%
\providecommand \bibitemNoStop [0]{.\EOS\space}%
\providecommand \EOS [0]{\spacefactor3000\relax}%
\providecommand \BibitemShut  [1]{\csname bibitem#1\endcsname}%
\let\auto@bib@innerbib\@empty
\bibitem [{\citenamefont {Larkin}\ and\ \citenamefont
  {Ovchinnikov}(1969)}]{LO69}%
  \BibitemOpen
  \bibfield  {author} {\bibinfo {author} {\bibfnamefont {A.~I.}\ \bibnamefont
  {Larkin}}\ and\ \bibinfo {author} {\bibfnamefont {Y.~N.}\ \bibnamefont
  {Ovchinnikov}},\ }\href
  {http://www.jetp.ac.ru/cgi-bin/e/index/e/28/6/p1200?a=list} {\bibfield
  {journal} {\bibinfo  {journal} {JETP}\ }\textbf {\bibinfo {volume} {28}},\
  \bibinfo {pages} {1200} (\bibinfo {year} {1969})}\BibitemShut {NoStop}%
\bibitem [{\citenamefont {Kitaev}(2014)}]{Kit14}%
  \BibitemOpen
  \bibfield  {author} {\bibinfo {author} {\bibfnamefont {A.}~\bibnamefont
  {Kitaev}},\ }in\ \href {http://online.kitp.ucsb.edu/online/joint98/kitaev/}
  {\emph {\bibinfo {booktitle} {talk given at the Fundamental Physics Prize
  Symposium}}}\ (\bibinfo {year} {2014})\BibitemShut {NoStop}%
\bibitem [{\citenamefont {Kitaev}(2015)}]{Kitaev}%
  \BibitemOpen
  \bibfield  {author} {\bibinfo {author} {\bibfnamefont {A.}~\bibnamefont
  {Kitaev}},\ }in\ \href
  {http://online.kitp.ucsb.edu/online/entangled15/kitaev/} {\emph {\bibinfo
  {booktitle} {talk given at KITP Program: Entanglement in Strongly-Correlated
  Quantum Matter}}}\ (\bibinfo {year} {2015})\BibitemShut {NoStop}%
\bibitem [{\citenamefont {Shenker}\ and\ \citenamefont
  {Stanford}(2014{\natexlab{a}})}]{Shenker2014}%
  \BibitemOpen
  \bibfield  {author} {\bibinfo {author} {\bibfnamefont {S.~H.}\ \bibnamefont
  {Shenker}}\ and\ \bibinfo {author} {\bibfnamefont {D.}~\bibnamefont
  {Stanford}},\ }\href {http://dx.doi.org/10.1007/JHEP03(2014)067} {\bibfield
  {journal} {\bibinfo  {journal} {J. High Energ. Phys.}\ }\textbf {\bibinfo
  {volume} {3}},\ \bibinfo {pages} {67} (\bibinfo {year}
  {2014}{\natexlab{a}})}\BibitemShut {NoStop}%
\bibitem [{\citenamefont {Shenker}\ and\ \citenamefont
  {Stanford}(2014{\natexlab{b}})}]{Shenker201412}%
  \BibitemOpen
  \bibfield  {author} {\bibinfo {author} {\bibfnamefont {S.~H.}\ \bibnamefont
  {Shenker}}\ and\ \bibinfo {author} {\bibfnamefont {D.}~\bibnamefont
  {Stanford}},\ }\href {http://dx.doi.org/10.1007/JHEP12(2014)046} {\bibfield
  {journal} {\bibinfo  {journal} {J. High Energ. Phys.}\ }\textbf {\bibinfo
  {volume} {12}},\ \bibinfo {pages} {46} (\bibinfo {year}
  {2014}{\natexlab{b}})}\BibitemShut {NoStop}%
\bibitem [{\citenamefont {Shenker}\ and\ \citenamefont
  {Stanford}(2015)}]{Shenker2015}%
  \BibitemOpen
  \bibfield  {author} {\bibinfo {author} {\bibfnamefont {S.~H.}\ \bibnamefont
  {Shenker}}\ and\ \bibinfo {author} {\bibfnamefont {D.}~\bibnamefont
  {Stanford}},\ }\href
  {https://link.springer.com/article/10.1007/JHEP05(2015)132} {\bibfield
  {journal} {\bibinfo  {journal} {J. High Energ. Phys.}\ }\textbf {\bibinfo
  {volume} {5}},\ \bibinfo {pages} {132} (\bibinfo {year} {2015})}\BibitemShut
  {NoStop}%
\bibitem [{\citenamefont {Roberts}\ and\ \citenamefont
  {Stanford}(2015)}]{Roberts:2014ab}%
  \BibitemOpen
  \bibfield  {author} {\bibinfo {author} {\bibfnamefont {D.~A.}\ \bibnamefont
  {Roberts}}\ and\ \bibinfo {author} {\bibfnamefont {D.}~\bibnamefont
  {Stanford}},\ }\href
  {http://journals.aps.org/prl/abstract/10.1103/PhysRevLett.115.131603}
  {\bibfield  {journal} {\bibinfo  {journal} {Phys. Rev. Lett.}\ }\textbf
  {\bibinfo {volume} {115}},\ \bibinfo {pages} {131603} (\bibinfo {year}
  {2015})}\BibitemShut {NoStop}%
\bibitem [{\citenamefont {Roberts}\ \emph {et~al.}(2015)\citenamefont
  {Roberts}, \citenamefont {Stanford},\ and\ \citenamefont
  {Susskind}}]{Roberts2015}%
  \BibitemOpen
  \bibfield  {author} {\bibinfo {author} {\bibfnamefont {D.~A.}\ \bibnamefont
  {Roberts}}, \bibinfo {author} {\bibfnamefont {D.}~\bibnamefont {Stanford}}, \
  and\ \bibinfo {author} {\bibfnamefont {L.}~\bibnamefont {Susskind}},\ }\href
  {http://dx.doi.org/10.1007/JHEP03(2015)051} {\bibfield  {journal} {\bibinfo
  {journal} {J. High Energ. Phys.}\ }\textbf {\bibinfo {volume} {3}},\ \bibinfo
  {pages} {51} (\bibinfo {year} {2015})}\BibitemShut {NoStop}%
\bibitem [{\citenamefont {Roberts}\ and\ \citenamefont
  {Swingle}(2016)}]{Roberts:2016aa}%
  \BibitemOpen
  \bibfield  {author} {\bibinfo {author} {\bibfnamefont {D.~A.}\ \bibnamefont
  {Roberts}}\ and\ \bibinfo {author} {\bibfnamefont {B.}~\bibnamefont
  {Swingle}},\ }\href
  {http://journals.aps.org/prl/abstract/10.1103/PhysRevLett.117.091602}
  {\bibfield  {journal} {\bibinfo  {journal} {Phys. Rev. Lett.}\ }\textbf
  {\bibinfo {volume} {117}},\ \bibinfo {pages} {091602} (\bibinfo {year}
  {2016})}\BibitemShut {NoStop}%
\bibitem [{\citenamefont {Hosur}\ \emph {et~al.}(2016)\citenamefont {Hosur},
  \citenamefont {Qi}, \citenamefont {Roberts},\ and\ \citenamefont
  {Yoshida}}]{Hosur2016}%
  \BibitemOpen
  \bibfield  {author} {\bibinfo {author} {\bibfnamefont {P.}~\bibnamefont
  {Hosur}}, \bibinfo {author} {\bibfnamefont {X.-L.}\ \bibnamefont {Qi}},
  \bibinfo {author} {\bibfnamefont {D.~A.}\ \bibnamefont {Roberts}}, \ and\
  \bibinfo {author} {\bibfnamefont {B.}~\bibnamefont {Yoshida}},\ }\href
  {http://dx.doi.org/10.1007/JHEP02(2016)004} {\bibfield  {journal} {\bibinfo
  {journal} {J. High Energ. Phys.}\ }\textbf {\bibinfo {volume} {2}},\ \bibinfo
  {pages} {4} (\bibinfo {year} {2016})}\BibitemShut {NoStop}%
\bibitem [{\citenamefont {Maldacena}\ \emph {et~al.}(2016)\citenamefont
  {Maldacena}, \citenamefont {Shenker},\ and\ \citenamefont
  {Stanford}}]{Maldacena2016}%
  \BibitemOpen
  \bibfield  {author} {\bibinfo {author} {\bibfnamefont {J.}~\bibnamefont
  {Maldacena}}, \bibinfo {author} {\bibfnamefont {S.~H.}\ \bibnamefont
  {Shenker}}, \ and\ \bibinfo {author} {\bibfnamefont {D.}~\bibnamefont
  {Stanford}},\ }\href {http://dx.doi.org/10.1007/JHEP08(2016)106} {\bibfield
  {journal} {\bibinfo  {journal} {J. High Energ. Phys.}\ }\textbf {\bibinfo
  {volume} {8}},\ \bibinfo {pages} {106} (\bibinfo {year} {2016})}\BibitemShut
  {NoStop}%
\bibitem [{\citenamefont {Polchinski}\ and\ \citenamefont
  {Rosenhaus}(2016)}]{polchinski2016spectrum}%
  \BibitemOpen
  \bibfield  {author} {\bibinfo {author} {\bibfnamefont {J.}~\bibnamefont
  {Polchinski}}\ and\ \bibinfo {author} {\bibfnamefont {V.}~\bibnamefont
  {Rosenhaus}},\ }\href {https://doi.org/10.1007/JHEP04(2016)001} {\bibfield
  {journal} {\bibinfo  {journal} {J. High Energ. Phys.}\ }\textbf {\bibinfo
  {volume} {4}},\ \bibinfo {pages} {1} (\bibinfo {year} {2016})}\BibitemShut
  {NoStop}%
\bibitem [{\citenamefont {Gu}\ \emph {et~al.}(2016)\citenamefont {Gu},
  \citenamefont {Qi},\ and\ \citenamefont {Stanford}}]{gu2016local}%
  \BibitemOpen
  \bibfield  {author} {\bibinfo {author} {\bibfnamefont {Y.}~\bibnamefont
  {Gu}}, \bibinfo {author} {\bibfnamefont {X.-L.}\ \bibnamefont {Qi}}, \ and\
  \bibinfo {author} {\bibfnamefont {D.}~\bibnamefont {Stanford}},\ }\href
  {https://doi.org/10.1007/JHEP05(2017)125} {\bibfield  {journal} {\bibinfo
  {journal} {J. High Energ. Phys.}\ }\textbf {\bibinfo {volume} {5}},\ \bibinfo
  {pages} {125} (\bibinfo {year} {2016})}\BibitemShut {NoStop}%
\bibitem [{\citenamefont {Rozenbaum}\ \emph {et~al.}(2017)\citenamefont
  {Rozenbaum}, \citenamefont {Ganeshan},\ and\ \citenamefont
  {Galitski}}]{PRL118.086801}%
  \BibitemOpen
  \bibfield  {author} {\bibinfo {author} {\bibfnamefont {E.~B.}\ \bibnamefont
  {Rozenbaum}}, \bibinfo {author} {\bibfnamefont {S.}~\bibnamefont {Ganeshan}},
  \ and\ \bibinfo {author} {\bibfnamefont {V.}~\bibnamefont {Galitski}},\
  }\href {\doibase 10.1103/PhysRevLett.118.086801} {\bibfield  {journal}
  {\bibinfo  {journal} {Phys. Rev. Lett.}\ }\textbf {\bibinfo {volume} {118}},\
  \bibinfo {pages} {086801} (\bibinfo {year} {2017})}\BibitemShut {NoStop}%
\bibitem [{\citenamefont {Roberts}\ and\ \citenamefont
  {Yoshida}(2016)}]{roberts2016chaos}%
  \BibitemOpen
  \bibfield  {author} {\bibinfo {author} {\bibfnamefont {D.~A.}\ \bibnamefont
  {Roberts}}\ and\ \bibinfo {author} {\bibfnamefont {B.}~\bibnamefont
  {Yoshida}},\ }\href {https://doi.org/10.1007/JHEP04(2017)121} {\bibfield
  {journal} {\bibinfo  {journal} {J. High Energ. Phys.}\ }\textbf {\bibinfo
  {volume} {4}},\ \bibinfo {pages} {121} (\bibinfo {year} {2016})}\BibitemShut
  {NoStop}%
\bibitem [{\citenamefont {Patel}\ \emph {et~al.}(2017)\citenamefont {Patel},
  \citenamefont {Chowdhury}, \citenamefont {Sachdev},\ and\ \citenamefont
  {Swingle}}]{patel2017}%
  \BibitemOpen
  \bibfield  {author} {\bibinfo {author} {\bibfnamefont {A.~A.}\ \bibnamefont
  {Patel}}, \bibinfo {author} {\bibfnamefont {D.}~\bibnamefont {Chowdhury}},
  \bibinfo {author} {\bibfnamefont {S.}~\bibnamefont {Sachdev}}, \ and\
  \bibinfo {author} {\bibfnamefont {B.}~\bibnamefont {Swingle}},\ }\href
  {https://journals.aps.org/prx/abstract/10.1103/PhysRevX.7.031047} {\bibfield
  {journal} {\bibinfo  {journal} {Phys. Rev. X}\ }\textbf {\bibinfo {volume}
  {7}},\ \bibinfo {pages} {031047} (\bibinfo {year} {2017})}\BibitemShut
  {NoStop}%
\bibitem [{\citenamefont {Huang}\ \emph
  {et~al.}(2017{\natexlab{a}})\citenamefont {Huang}, \citenamefont {Brandao},\
  and\ \citenamefont {Zhang}}]{huang2017}%
  \BibitemOpen
  \bibfield  {author} {\bibinfo {author} {\bibfnamefont {Y.}~\bibnamefont
  {Huang}}, \bibinfo {author} {\bibfnamefont {F.~G.}\ \bibnamefont {Brandao}},
  \ and\ \bibinfo {author} {\bibfnamefont {Y.-L.}\ \bibnamefont {Zhang}},\
  }\href {https://arxiv.org/abs/1705.07597} {\bibfield  {journal} {\bibinfo
  {journal} {arXiv: 1705.07597}\ } (\bibinfo {year}
  {2017}{\natexlab{a}})}\BibitemShut {NoStop}%
\bibitem [{\citenamefont {Mezei}\ and\ \citenamefont
  {Stanford}(2017)}]{Mezei2017}%
  \BibitemOpen
  \bibfield  {author} {\bibinfo {author} {\bibfnamefont {M.}~\bibnamefont
  {Mezei}}\ and\ \bibinfo {author} {\bibfnamefont {D.}~\bibnamefont
  {Stanford}},\ }\href {https://doi.org/10.1007/JHEP05(2017)065} {\bibfield
  {journal} {\bibinfo  {journal} {J. High Energ. Phys.}\ }\textbf {\bibinfo
  {volume} {5}},\ \bibinfo {pages} {65} (\bibinfo {year} {2017})}\BibitemShut
  {NoStop}%
\bibitem [{\citenamefont {Lucas}(2017)}]{Lucas:2017aa}%
  \BibitemOpen
  \bibfield  {author} {\bibinfo {author} {\bibfnamefont {A.}~\bibnamefont
  {Lucas}},\ }\href {https://arxiv.org/abs/1710.01005} {\bibfield  {journal}
  {\bibinfo  {journal} {arXiv: 1710.01005}\ } (\bibinfo {year}
  {2017})}\BibitemShut {NoStop}%
\bibitem [{\citenamefont {Huang}\ \emph
  {et~al.}(2017{\natexlab{b}})\citenamefont {Huang}, \citenamefont {Zhang},\
  and\ \citenamefont {Chen}}]{huang2016}%
  \BibitemOpen
  \bibfield  {author} {\bibinfo {author} {\bibfnamefont {Y.}~\bibnamefont
  {Huang}}, \bibinfo {author} {\bibfnamefont {Y.-L.}\ \bibnamefont {Zhang}}, \
  and\ \bibinfo {author} {\bibfnamefont {X.}~\bibnamefont {Chen}},\ }\href
  {\doibase 10.1002/andp.201600318} {\bibfield  {journal} {\bibinfo  {journal}
  {Ann. Phys. (Berl.)}\ }\textbf {\bibinfo {volume} {529}},\ \bibinfo {pages}
  {1600318} (\bibinfo {year} {2017}{\natexlab{b}})}\BibitemShut {NoStop}%
\bibitem [{\citenamefont {{Fan}}\ \emph {et~al.}(2016)\citenamefont {{Fan}},
  \citenamefont {{Zhang}}, \citenamefont {{Shen}},\ and\ \citenamefont
  {{Zhai}}}]{FZSZ16}%
  \BibitemOpen
  \bibfield  {author} {\bibinfo {author} {\bibfnamefont {R.}~\bibnamefont
  {{Fan}}}, \bibinfo {author} {\bibfnamefont {P.}~\bibnamefont {{Zhang}}},
  \bibinfo {author} {\bibfnamefont {H.}~\bibnamefont {{Shen}}}, \ and\ \bibinfo
  {author} {\bibfnamefont {H.}~\bibnamefont {{Zhai}}},\ }\href
  {https://doi.org/10.1016/j.scib.2017.04.011} {\bibfield  {journal} {\bibinfo
  {journal} {Sci. Bull.}\ }\textbf {\bibinfo {volume} {62}},\ \bibinfo {pages}
  {707} (\bibinfo {year} {2016})}\BibitemShut {NoStop}%
\bibitem [{\citenamefont {{Chen}}(2016)}]{C16}%
  \BibitemOpen
  \bibfield  {author} {\bibinfo {author} {\bibfnamefont {Y.}~\bibnamefont
  {{Chen}}},\ }\href {https://arxiv.org/abs/1608.02765} {\bibfield  {journal}
  {\bibinfo  {journal} {arXiv:1608.02765}\ } (\bibinfo {year}
  {2016})}\BibitemShut {NoStop}%
\bibitem [{\citenamefont {{Swingle}}\ and\ \citenamefont
  {{Chowdhury}}(2017)}]{SC16}%
  \BibitemOpen
  \bibfield  {author} {\bibinfo {author} {\bibfnamefont {B.}~\bibnamefont
  {{Swingle}}}\ and\ \bibinfo {author} {\bibfnamefont {D.}~\bibnamefont
  {{Chowdhury}}},\ }\href
  {https://journals.aps.org/prb/abstract/10.1103/PhysRevB.95.060201} {\bibfield
   {journal} {\bibinfo  {journal} {Phys. Rev. B}\ }\textbf {\bibinfo {volume}
  {95}},\ \bibinfo {pages} {060201(R)} (\bibinfo {year} {2017})}\BibitemShut
  {NoStop}%
\bibitem [{\citenamefont {{He}}\ and\ \citenamefont {{Lu}}(2017)}]{HL16}%
  \BibitemOpen
  \bibfield  {author} {\bibinfo {author} {\bibfnamefont {R.-Q.}\ \bibnamefont
  {{He}}}\ and\ \bibinfo {author} {\bibfnamefont {Z.-Y.}\ \bibnamefont
  {{Lu}}},\ }\href
  {http://journals.aps.org/prb/abstract/10.1103/PhysRevB.95.054201} {\bibfield
  {journal} {\bibinfo  {journal} {Phys. Rev. B}\ }\textbf {\bibinfo {volume}
  {95}},\ \bibinfo {pages} {054201} (\bibinfo {year} {2017})}\BibitemShut
  {NoStop}%
\bibitem [{\citenamefont {Chen}\ \emph {et~al.}(2017)\citenamefont {Chen},
  \citenamefont {Zhou}, \citenamefont {Huse},\ and\ \citenamefont
  {Fradkin}}]{chen2016}%
  \BibitemOpen
  \bibfield  {author} {\bibinfo {author} {\bibfnamefont {X.}~\bibnamefont
  {Chen}}, \bibinfo {author} {\bibfnamefont {T.}~\bibnamefont {Zhou}}, \bibinfo
  {author} {\bibfnamefont {D.~A.}\ \bibnamefont {Huse}}, \ and\ \bibinfo
  {author} {\bibfnamefont {E.}~\bibnamefont {Fradkin}},\ }\href {\doibase
  10.1002/andp.201600332} {\bibfield  {journal} {\bibinfo  {journal} {Ann.
  Phys. (Berl.)}\ }\textbf {\bibinfo {volume} {529}},\ \bibinfo {pages}
  {1600332} (\bibinfo {year} {2017})}\BibitemShut {NoStop}%
\bibitem [{\citenamefont {Slagle}\ \emph {et~al.}(2017)\citenamefont {Slagle},
  \citenamefont {Bi}, \citenamefont {You},\ and\ \citenamefont
  {Xu}}]{slagle2017out}%
  \BibitemOpen
  \bibfield  {author} {\bibinfo {author} {\bibfnamefont {K.}~\bibnamefont
  {Slagle}}, \bibinfo {author} {\bibfnamefont {Z.}~\bibnamefont {Bi}}, \bibinfo
  {author} {\bibfnamefont {Y.-Z.}\ \bibnamefont {You}}, \ and\ \bibinfo
  {author} {\bibfnamefont {C.}~\bibnamefont {Xu}},\ }\href
  {https://journals.aps.org/prb/abstract/10.1103/PhysRevB.95.165136} {\bibfield
   {journal} {\bibinfo  {journal} {Phys. Rev. B}\ }\textbf {\bibinfo {volume}
  {95}},\ \bibinfo {pages} {165136} (\bibinfo {year} {2017})}\BibitemShut
  {NoStop}%
\bibitem [{\citenamefont {Song}\ \emph {et~al.}(2017)\citenamefont {Song},
  \citenamefont {Jian},\ and\ \citenamefont {Balents}}]{Song2017}%
  \BibitemOpen
  \bibfield  {author} {\bibinfo {author} {\bibfnamefont {X.-Y.}\ \bibnamefont
  {Song}}, \bibinfo {author} {\bibfnamefont {C.-M.}\ \bibnamefont {Jian}}, \
  and\ \bibinfo {author} {\bibfnamefont {L.}~\bibnamefont {Balents}},\ }\href
  {\doibase 10.1103/PhysRevLett.119.216601} {\bibfield  {journal} {\bibinfo
  {journal} {Phys. Rev. Lett.}\ }\textbf {\bibinfo {volume} {119}},\ \bibinfo
  {pages} {216601} (\bibinfo {year} {2017})}\BibitemShut {NoStop}%
\bibitem [{\citenamefont {{Ben-Zion}}\ and\ \citenamefont
  {{McGreevy}}(2017)}]{Ben-Zion2017}%
  \BibitemOpen
  \bibfield  {author} {\bibinfo {author} {\bibfnamefont {D.}~\bibnamefont
  {{Ben-Zion}}}\ and\ \bibinfo {author} {\bibfnamefont {J.}~\bibnamefont
  {{McGreevy}}},\ }\href {https://arxiv.org/abs/1711.02686} {\bibfield
  {journal} {\bibinfo  {journal} {arXiv: 1711.02686}\ } (\bibinfo {year}
  {2017})}\BibitemShut {NoStop}%
\bibitem [{\citenamefont {{Xu}}\ and\ \citenamefont
  {{Swingle}}(2018)}]{Xu2018}%
  \BibitemOpen
  \bibfield  {author} {\bibinfo {author} {\bibfnamefont {S.}~\bibnamefont
  {{Xu}}}\ and\ \bibinfo {author} {\bibfnamefont {B.}~\bibnamefont
  {{Swingle}}},\ }\href {https://arxiv.org/abs/1802.00801} {\bibfield
  {journal} {\bibinfo  {journal} {arXiv: 1802.00801}\ } (\bibinfo {year}
  {2018})}\BibitemShut {NoStop}%
\bibitem [{\citenamefont {Swingle}\ \emph {et~al.}(2016)\citenamefont
  {Swingle}, \citenamefont {Bentsen}, \citenamefont {Schleier-Smith},\ and\
  \citenamefont {Hayden}}]{Swingle:2016aa}%
  \BibitemOpen
  \bibfield  {author} {\bibinfo {author} {\bibfnamefont {B.}~\bibnamefont
  {Swingle}}, \bibinfo {author} {\bibfnamefont {G.}~\bibnamefont {Bentsen}},
  \bibinfo {author} {\bibfnamefont {M.}~\bibnamefont {Schleier-Smith}}, \ and\
  \bibinfo {author} {\bibfnamefont {P.}~\bibnamefont {Hayden}},\ }\href
  {http://journals.aps.org/pra/abstract/10.1103/PhysRevA.94.040302} {\bibfield
  {journal} {\bibinfo  {journal} {Phys. Rev. A}\ }\textbf {\bibinfo {volume}
  {94}},\ \bibinfo {pages} {040302} (\bibinfo {year} {2016})}\BibitemShut
  {NoStop}%
\bibitem [{\citenamefont {{Zhu}}\ \emph {et~al.}(2016)\citenamefont {{Zhu}},
  \citenamefont {{Hafezi}},\ and\ \citenamefont {{Grover}}}]{ZHG16}%
  \BibitemOpen
  \bibfield  {author} {\bibinfo {author} {\bibfnamefont {G.}~\bibnamefont
  {{Zhu}}}, \bibinfo {author} {\bibfnamefont {M.}~\bibnamefont {{Hafezi}}}, \
  and\ \bibinfo {author} {\bibfnamefont {T.}~\bibnamefont {{Grover}}},\ }\href
  {https://journals.aps.org/pra/abstract/10.1103/PhysRevA.94.062329} {\bibfield
   {journal} {\bibinfo  {journal} {Phys. Rev. A}\ }\textbf {\bibinfo {volume}
  {94}},\ \bibinfo {pages} {062329} (\bibinfo {year} {2016})}\BibitemShut
  {NoStop}%
\bibitem [{\citenamefont {Yao}\ \emph {et~al.}(2016)\citenamefont {Yao},
  \citenamefont {Grusdt}, \citenamefont {Swingle}, \citenamefont {Lukin},
  \citenamefont {Stamper-Kurn}, \citenamefont {Moore},\ and\ \citenamefont
  {Demler}}]{YGS+16}%
  \BibitemOpen
  \bibfield  {author} {\bibinfo {author} {\bibfnamefont {N.~Y.}\ \bibnamefont
  {Yao}}, \bibinfo {author} {\bibfnamefont {F.}~\bibnamefont {Grusdt}},
  \bibinfo {author} {\bibfnamefont {B.}~\bibnamefont {Swingle}}, \bibinfo
  {author} {\bibfnamefont {M.~D.}\ \bibnamefont {Lukin}}, \bibinfo {author}
  {\bibfnamefont {D.~M.}\ \bibnamefont {Stamper-Kurn}}, \bibinfo {author}
  {\bibfnamefont {J.~E.}\ \bibnamefont {Moore}}, \ and\ \bibinfo {author}
  {\bibfnamefont {E.~A.}\ \bibnamefont {Demler}},\ }\href
  {https://arxiv.org/abs/1607.01801} {\bibfield  {journal} {\bibinfo  {journal}
  {arXiv: 1607.01801}\ } (\bibinfo {year} {2016})}\BibitemShut {NoStop}%
\bibitem [{\citenamefont {Tsuji}\ \emph {et~al.}(2017)\citenamefont {Tsuji},
  \citenamefont {Werner},\ and\ \citenamefont {Ueda}}]{tsuji2017exact}%
  \BibitemOpen
  \bibfield  {author} {\bibinfo {author} {\bibfnamefont {N.}~\bibnamefont
  {Tsuji}}, \bibinfo {author} {\bibfnamefont {P.}~\bibnamefont {Werner}}, \
  and\ \bibinfo {author} {\bibfnamefont {M.}~\bibnamefont {Ueda}},\ }\href
  {https://journals.aps.org/pra/abstract/10.1103/PhysRevA.95.011601} {\bibfield
   {journal} {\bibinfo  {journal} {Phys. Rev. A}\ }\textbf {\bibinfo {volume}
  {95}},\ \bibinfo {pages} {011601} (\bibinfo {year} {2017})}\BibitemShut
  {NoStop}%
\bibitem [{\citenamefont {Bohrdt}\ \emph {et~al.}(2017)\citenamefont {Bohrdt},
  \citenamefont {Mendl}, \citenamefont {Endres},\ and\ \citenamefont
  {Knap}}]{bohrdt2016}%
  \BibitemOpen
  \bibfield  {author} {\bibinfo {author} {\bibfnamefont {A.}~\bibnamefont
  {Bohrdt}}, \bibinfo {author} {\bibfnamefont {C.}~\bibnamefont {Mendl}},
  \bibinfo {author} {\bibfnamefont {M.}~\bibnamefont {Endres}}, \ and\ \bibinfo
  {author} {\bibfnamefont {M.}~\bibnamefont {Knap}},\ }\href
  {http://iopscience.iop.org/article/10.1088/1367-2630/aa719b} {\bibfield
  {journal} {\bibinfo  {journal} {New J. Phys.}\ }\textbf {\bibinfo {volume}
  {19}},\ \bibinfo {pages} {063001} (\bibinfo {year} {2017})}\BibitemShut
  {NoStop}%
\bibitem [{\citenamefont {G{\"a}rttner}\ \emph {et~al.}(2017)\citenamefont
  {G{\"a}rttner}, \citenamefont {Bohnet}, \citenamefont {Safavi-Naini},
  \citenamefont {Wall}, \citenamefont {Bollinger},\ and\ \citenamefont
  {Rey}}]{garttner2016}%
  \BibitemOpen
  \bibfield  {author} {\bibinfo {author} {\bibfnamefont {M.}~\bibnamefont
  {G{\"a}rttner}}, \bibinfo {author} {\bibfnamefont {J.~G.}\ \bibnamefont
  {Bohnet}}, \bibinfo {author} {\bibfnamefont {A.}~\bibnamefont
  {Safavi-Naini}}, \bibinfo {author} {\bibfnamefont {M.~L.}\ \bibnamefont
  {Wall}}, \bibinfo {author} {\bibfnamefont {J.~J.}\ \bibnamefont {Bollinger}},
  \ and\ \bibinfo {author} {\bibfnamefont {A.~M.}\ \bibnamefont {Rey}},\ }\href
  {https://www.nature.com/nphys/journal/v13/n8/full/nphys4119.html} {\bibfield
  {journal} {\bibinfo  {journal} {Nat. Phys.}\ }\textbf {\bibinfo {volume}
  {13}},\ \bibinfo {pages} {781} (\bibinfo {year} {2017})}\BibitemShut
  {NoStop}%
\bibitem [{\citenamefont {Li}\ \emph {et~al.}(2017)\citenamefont {Li},
  \citenamefont {Fan}, \citenamefont {Wang}, \citenamefont {Ye}, \citenamefont
  {Zeng}, \citenamefont {Zhai}, \citenamefont {Peng},\ and\ \citenamefont
  {Du}}]{li2016measuring}%
  \BibitemOpen
  \bibfield  {author} {\bibinfo {author} {\bibfnamefont {J.}~\bibnamefont
  {Li}}, \bibinfo {author} {\bibfnamefont {R.}~\bibnamefont {Fan}}, \bibinfo
  {author} {\bibfnamefont {H.}~\bibnamefont {Wang}}, \bibinfo {author}
  {\bibfnamefont {B.}~\bibnamefont {Ye}}, \bibinfo {author} {\bibfnamefont
  {B.}~\bibnamefont {Zeng}}, \bibinfo {author} {\bibfnamefont {H.}~\bibnamefont
  {Zhai}}, \bibinfo {author} {\bibfnamefont {X.}~\bibnamefont {Peng}}, \ and\
  \bibinfo {author} {\bibfnamefont {J.}~\bibnamefont {Du}},\ }\href
  {https://journals.aps.org/prx/abstract/10.1103/PhysRevX.7.031011} {\bibfield
  {journal} {\bibinfo  {journal} {Phys. Rev. X}\ } (\bibinfo {year}
  {2017})}\BibitemShut {NoStop}%
\bibitem [{\citenamefont {Halpern}(2017)}]{Halpern:2017ab}%
  \BibitemOpen
  \bibfield  {author} {\bibinfo {author} {\bibfnamefont {N.~Y.}\ \bibnamefont
  {Halpern}},\ }\href {\doibase 10.1103/PhysRevA.95.012120} {\bibfield
  {journal} {\bibinfo  {journal} {Phys. Rev. A}\ }\textbf {\bibinfo {volume}
  {95}},\ \bibinfo {pages} {012120} (\bibinfo {year} {2017})}\BibitemShut
  {NoStop}%
\bibitem [{\citenamefont {Halpern}\ \emph {et~al.}(2017)\citenamefont
  {Halpern}, \citenamefont {Swingle},\ and\ \citenamefont
  {Dressel}}]{Halpern:2017aa}%
  \BibitemOpen
  \bibfield  {author} {\bibinfo {author} {\bibfnamefont {N.~Y.}\ \bibnamefont
  {Halpern}}, \bibinfo {author} {\bibfnamefont {B.}~\bibnamefont {Swingle}}, \
  and\ \bibinfo {author} {\bibfnamefont {J.}~\bibnamefont {Dressel}},\ }\href
  {https://arxiv.org/abs/1704.01971} {\bibfield  {journal} {\bibinfo  {journal}
  {arXiv: 1704.01971}\ } (\bibinfo {year} {2017})}\BibitemShut {NoStop}%
\bibitem [{\citenamefont {Ba\~nuls}\ \emph {et~al.}(2011)\citenamefont
  {Ba\~nuls}, \citenamefont {Cirac},\ and\ \citenamefont
  {Hastings}}]{PRL106.050405}%
  \BibitemOpen
  \bibfield  {author} {\bibinfo {author} {\bibfnamefont {M.~C.}\ \bibnamefont
  {Ba\~nuls}}, \bibinfo {author} {\bibfnamefont {J.~I.}\ \bibnamefont {Cirac}},
  \ and\ \bibinfo {author} {\bibfnamefont {M.~B.}\ \bibnamefont {Hastings}},\
  }\href {\doibase 10.1103/PhysRevLett.106.050405} {\bibfield  {journal}
  {\bibinfo  {journal} {Phys. Rev. Lett.}\ }\textbf {\bibinfo {volume} {106}},\
  \bibinfo {pages} {050405} (\bibinfo {year} {2011})}\BibitemShut {NoStop}%
\bibitem [{\citenamefont {Breuer}\ and\ \citenamefont
  {Petruccione}(2002)}]{breuer2002}%
  \BibitemOpen
  \bibfield  {author} {\bibinfo {author} {\bibfnamefont {H.-P.}\ \bibnamefont
  {Breuer}}\ and\ \bibinfo {author} {\bibfnamefont {F.}~\bibnamefont
  {Petruccione}},\ }\href
  {http://www.oxfordscholarship.com/view/10.1093/acprof:oso/9780199213900.001.0001/acprof-9780199213900}
  {\emph {\bibinfo {title} {The theory of open quantum systems}}}\ (\bibinfo
  {publisher} {Oxford University Press},\ \bibinfo {year} {2002})\BibitemShut
  {NoStop}%
\bibitem [{\citenamefont {Zeng}\ \emph {et~al.}(2015)\citenamefont {Zeng},
  \citenamefont {Chen}, \citenamefont {Zhou},\ and\ \citenamefont
  {Wen}}]{zeng2015}%
  \BibitemOpen
  \bibfield  {author} {\bibinfo {author} {\bibfnamefont {B.}~\bibnamefont
  {Zeng}}, \bibinfo {author} {\bibfnamefont {X.}~\bibnamefont {Chen}}, \bibinfo
  {author} {\bibfnamefont {D.-L.}\ \bibnamefont {Zhou}}, \ and\ \bibinfo
  {author} {\bibfnamefont {X.-G.}\ \bibnamefont {Wen}},\ }\href
  {https://arxiv.org/abs/1508.02595} {\bibfield  {journal} {\bibinfo  {journal}
  {arXiv:1508.02595}\ } (\bibinfo {year} {2015})}\BibitemShut {NoStop}%
\bibitem [{\citenamefont {Johansson}\ \emph {et~al.}(2012)\citenamefont
  {Johansson}, \citenamefont {Nation},\ and\ \citenamefont {Nori}}]{qutip1}%
  \BibitemOpen
  \bibfield  {author} {\bibinfo {author} {\bibfnamefont {J.}~\bibnamefont
  {Johansson}}, \bibinfo {author} {\bibfnamefont {P.}~\bibnamefont {Nation}}, \
  and\ \bibinfo {author} {\bibfnamefont {F.}~\bibnamefont {Nori}},\ }\href
  {http://www.sciencedirect.com/science/article/pii/S0010465512000835}
  {\bibfield  {journal} {\bibinfo  {journal} {Comput. Phys. Commun.}\ }\textbf
  {\bibinfo {volume} {183}},\ \bibinfo {pages} {1760 } (\bibinfo {year}
  {2012})}\BibitemShut {NoStop}%
\bibitem [{\citenamefont {Johansson}\ \emph {et~al.}(2013)\citenamefont
  {Johansson}, \citenamefont {Nation},\ and\ \citenamefont {Nori}}]{qutip2}%
  \BibitemOpen
  \bibfield  {author} {\bibinfo {author} {\bibfnamefont {J.}~\bibnamefont
  {Johansson}}, \bibinfo {author} {\bibfnamefont {P.}~\bibnamefont {Nation}}, \
  and\ \bibinfo {author} {\bibfnamefont {F.}~\bibnamefont {Nori}},\ }\href
  {http://www.sciencedirect.com/science/article/pii/S0010465512003955}
  {\bibfield  {journal} {\bibinfo  {journal} {Comput. Phys. Commun.}\ }\textbf
  {\bibinfo {volume} {184}},\ \bibinfo {pages} {1234 } (\bibinfo {year}
  {2013})}\BibitemShut {NoStop}%
\bibitem [{\citenamefont {Vidal}(2004)}]{vidal2004efficient}%
  \BibitemOpen
  \bibfield  {author} {\bibinfo {author} {\bibfnamefont {G.}~\bibnamefont
  {Vidal}},\ }\href
  {https://journals.aps.org/prl/pdf/10.1103/PhysRevLett.93.040502} {\bibfield
  {journal} {\bibinfo  {journal} {Phys. Rev. Lett.}\ }\textbf {\bibinfo
  {volume} {93}},\ \bibinfo {pages} {040502} (\bibinfo {year}
  {2004})}\BibitemShut {NoStop}%
\bibitem [{\citenamefont {Zwolak}\ and\ \citenamefont
  {Vidal}(2004)}]{zwolak2004mixed}%
  \BibitemOpen
  \bibfield  {author} {\bibinfo {author} {\bibfnamefont {M.}~\bibnamefont
  {Zwolak}}\ and\ \bibinfo {author} {\bibfnamefont {G.}~\bibnamefont {Vidal}},\
  }\href {https://journals.aps.org/prl/pdf/10.1103/PhysRevLett.93.207205}
  {\bibfield  {journal} {\bibinfo  {journal} {Phys. Rev. Lett.}\ }\textbf
  {\bibinfo {volume} {93}},\ \bibinfo {pages} {207205} (\bibinfo {year}
  {2004})}\BibitemShut {NoStop}%
\bibitem [{\citenamefont {Schollw{\"o}ck}(2011)}]{schollwock2011}%
  \BibitemOpen
  \bibfield  {author} {\bibinfo {author} {\bibfnamefont {U.}~\bibnamefont
  {Schollw{\"o}ck}},\ }\href
  {http://www.sciencedirect.com/science/article/pii/S0003491610001752}
  {\bibfield  {journal} {\bibinfo  {journal} {Ann. Phys.}\ }\textbf {\bibinfo
  {volume} {326}},\ \bibinfo {pages} {96} (\bibinfo {year} {2011})}\BibitemShut
  {NoStop}%
\bibitem [{\citenamefont {Lieb}\ and\ \citenamefont {Robinson}(1972)}]{LR72}%
  \BibitemOpen
  \bibfield  {author} {\bibinfo {author} {\bibfnamefont {E.~H.}\ \bibnamefont
  {Lieb}}\ and\ \bibinfo {author} {\bibfnamefont {D.~W.}\ \bibnamefont
  {Robinson}},\ }\href {\doibase 10.1007/BF01645779} {\bibfield  {journal}
  {\bibinfo  {journal} {Commun. Math. Phys.}\ }\textbf {\bibinfo {volume}
  {28}},\ \bibinfo {pages} {251} (\bibinfo {year} {1972})}\BibitemShut
  {NoStop}%
\bibitem [{\citenamefont {Hastings}\ and\ \citenamefont {Koma}(2006)}]{HK06}%
  \BibitemOpen
  \bibfield  {author} {\bibinfo {author} {\bibfnamefont {M.~B.}\ \bibnamefont
  {Hastings}}\ and\ \bibinfo {author} {\bibfnamefont {T.}~\bibnamefont
  {Koma}},\ }\href {\doibase 10.1007/s00220-006-0030-4} {\bibfield  {journal}
  {\bibinfo  {journal} {Commun. Math. Phys.}\ }\textbf {\bibinfo {volume}
  {265}},\ \bibinfo {pages} {781} (\bibinfo {year} {2006})}\BibitemShut
  {NoStop}%
\bibitem [{\citenamefont {Nachtergaele}\ \emph {et~al.}(2006)\citenamefont
  {Nachtergaele}, \citenamefont {Ogata},\ and\ \citenamefont
  {Sims}}]{NOS06propg}%
  \BibitemOpen
  \bibfield  {author} {\bibinfo {author} {\bibfnamefont {B.}~\bibnamefont
  {Nachtergaele}}, \bibinfo {author} {\bibfnamefont {Y.}~\bibnamefont {Ogata}},
  \ and\ \bibinfo {author} {\bibfnamefont {R.}~\bibnamefont {Sims}},\ }\href
  {http://link.springer.com/article/10.1007%2Fs10955-006-9143-6} {\bibfield
  {journal} {\bibinfo  {journal} {J. Stat. Phys.}\ }\textbf {\bibinfo {volume}
  {124}},\ \bibinfo {pages} {1} (\bibinfo {year} {2006})}\BibitemShut {NoStop}%
\bibitem [{\citenamefont {Poulin}(2010)}]{poulin2010lieb}%
  \BibitemOpen
  \bibfield  {author} {\bibinfo {author} {\bibfnamefont {D.}~\bibnamefont
  {Poulin}},\ }\href {\doibase 10.1103/PhysRevLett.104.190401} {\bibfield
  {journal} {\bibinfo  {journal} {Phys. Rev. Lett.}\ }\textbf {\bibinfo
  {volume} {104}},\ \bibinfo {pages} {190401} (\bibinfo {year}
  {2010})}\BibitemShut {NoStop}%
\bibitem [{\citenamefont {Nachtergaele}\ \emph {et~al.}(2011)\citenamefont
  {Nachtergaele}, \citenamefont {Vershynina},\ and\ \citenamefont
  {Zagrebnov}}]{NVZ2011lieb}%
  \BibitemOpen
  \bibfield  {author} {\bibinfo {author} {\bibfnamefont {B.}~\bibnamefont
  {Nachtergaele}}, \bibinfo {author} {\bibfnamefont {A.}~\bibnamefont
  {Vershynina}}, \ and\ \bibinfo {author} {\bibfnamefont {V.~A.}\ \bibnamefont
  {Zagrebnov}},\ }\href {https://arxiv.org/abs/1103.1122} {\bibfield  {journal}
  {\bibinfo  {journal} {AMS Contemp. Math.}\ }\textbf {\bibinfo {volume}
  {552}},\ \bibinfo {pages} {161} (\bibinfo {year} {2011})}\BibitemShut
  {NoStop}%
\bibitem [{\citenamefont {Barthel}\ and\ \citenamefont
  {Kliesch}(2012)}]{BK2012}%
  \BibitemOpen
  \bibfield  {author} {\bibinfo {author} {\bibfnamefont {T.}~\bibnamefont
  {Barthel}}\ and\ \bibinfo {author} {\bibfnamefont {M.}~\bibnamefont
  {Kliesch}},\ }\href {\doibase 10.1103/PhysRevLett.108.230504} {\bibfield
  {journal} {\bibinfo  {journal} {Phys. Rev. Lett.}\ }\textbf {\bibinfo
  {volume} {108}},\ \bibinfo {pages} {230504} (\bibinfo {year}
  {2012})}\BibitemShut {NoStop}%
\bibitem [{\citenamefont {Descamps}(2013)}]{descamps2013}%
  \BibitemOpen
  \bibfield  {author} {\bibinfo {author} {\bibfnamefont {B.}~\bibnamefont
  {Descamps}},\ }\href {\doibase 10.1063/1.4820785} {\bibfield  {journal}
  {\bibinfo  {journal} {J. Math. Phys.}\ }\textbf {\bibinfo {volume} {54}},\
  \bibinfo {pages} {092202} (\bibinfo {year} {2013})}\BibitemShut {NoStop}%
\bibitem [{\citenamefont {Kliesch}\ \emph {et~al.}(2014)\citenamefont
  {Kliesch}, \citenamefont {Gogolin},\ and\ \citenamefont
  {Eisert}}]{kliesch2014lieb}%
  \BibitemOpen
  \bibfield  {author} {\bibinfo {author} {\bibfnamefont {M.}~\bibnamefont
  {Kliesch}}, \bibinfo {author} {\bibfnamefont {C.}~\bibnamefont {Gogolin}}, \
  and\ \bibinfo {author} {\bibfnamefont {J.}~\bibnamefont {Eisert}},\ }in\
  \href {\doibase 10.1007/978-3-319-06379-9_17} {\emph {\bibinfo {booktitle}
  {Many-Electron Approaches in Physics, Chemistry and Mathematics}}}\ (\bibinfo
   {publisher} {Springer International Publishing},\ \bibinfo {year} {2014})\
  pp.\ \bibinfo {pages} {301--318}\BibitemShut {NoStop}%
\bibitem [{\citenamefont {Swingle}\ and\ \citenamefont
  {Halpern}(2018)}]{Swingle:2018aa}%
  \BibitemOpen
  \bibfield  {author} {\bibinfo {author} {\bibfnamefont {B.}~\bibnamefont
  {Swingle}}\ and\ \bibinfo {author} {\bibfnamefont {N.~Y.}\ \bibnamefont
  {Halpern}},\ }\href {https://arxiv.org/abs/1802.01587} {\bibfield  {journal}
  {\bibinfo  {journal} {arXiv: 1802.01587}\ } (\bibinfo {year}
  {2018})}\BibitemShut {NoStop}%
\end{thebibliography}

\end{document}